\documentclass[12pt]{article}
\usepackage[utf8]{inputenc}
\usepackage{amsmath, amsthm, amssymb, tikz}%, hyperref, enumerate}
\usepackage{amsfonts}
\usepackage{fullpage}

\usepackage{amsmath,color}
\usepackage[enableskew,vcentermath]{youngtab}
\RequirePackage{youngtab}
\RequirePackage{mathptmx}
\usepackage{amsmath,amssymb,color,graphics,graphicx}
\usepackage{bm}% bold math
\theoremstyle{plain}% default
\newtheorem{thm}{Theorem}%[section]

\newtheorem{prop}[thm]{Proposition}

\theoremstyle{definition}

\theoremstyle{remark}

\def\a{\alpha}

\def\b{\beta}

\def\d{\delta}

\def\vt{\vartheta}

\def\o{\omega}

\title{Decomposition of third-order constitutive tensors}% and some physical applications}
\author{Yakov Itin and Shulamit Reches\\
Mathematics Department, Jerusalem College of Technology, Israel}
\date{\today}

\usepackage{graphicx}

\begin{document}

\maketitle
\tableofcontents
\begin{abstract}
    Third-order tensors are widely used as a mathematical tool for modeling physical properties of media in solid state physics. In most cases, they arise as constitutive tensors of proportionality between basic physics quantities.   The constitutive tensor can be considered as a complete set of physical parameters of a medium. The algebraic features of the constitutive tensor can be seen as a tool for proper identification of natural material, as crystals, and for design the artificial nano-materials with  prescribed properties.

    In this paper, we study the algebraic properties of a generic 3-rd order tensor relative to its invariant decomposition. In a correspondence to different groups acted on the basic vector space, we present the hierarchy of types of tensor  decomposition  into invariant subtensors. In particular, we discuss the problem of non-uniqueness and reducibility of high-order tensor  decomposition. For a generic 3-rd order tensor, these features are described  explicitly. 
    In the case of special tensors of a prescribed  symmetry, the decomposition turns out to be irreducible and unique. We present the explicit results for two physically interesting models: the piezoelectric  tensor as an example of a pair  symmetry and the Hall tensor as an example of  a pair skew-symmetry. 
\end{abstract}
\section{Introduction}
Three-dimensional third-order tensors have a wide rang of applications especially in solid state physics, see e.g. \cite{Boehler}, \cite{Hartmann}, \cite{Haus}, \cite{Nye}, \cite{Tinder}. 
The well-known examples are the piezoelectric and piezomagnetic tensors that connect the electric and magnetic fields respectively to the elasticity stress tensor. The corresponded third-order  tensors are symmetric in a pair of their indices and thus have at most 18 independent components. 
The electro-optical tensor and the second harmonic generation  tensor in optics  relate  dielectric impermeability and  dielectric polarization respectively to the exterior electric field. These quantities are  third-order tensors of the same pair  symmetry type.  The Hall  tensor is a third-order tensor of an alternative algebraic symmetry. It is skew-symmetric in a pair of its indices and thus has at most 9 independent components. 

As it is well-known, the distinct components of a tensor themselves do not have an invariant meaning. Their magnitude depend on the coordinate system used. A  covariant meaning, however, can be given to different  subsets of the components that form  sub-tensors of a given tensor. The separation of the whole set of tensor  components into subsets that are tensors themselves is a useful procedure. It is usually termed as {\it irreducible decomposition}. In applications literature, 
this procedure  is used sometimes rather ambiguously, without given the precise definitions and the necessary conditions. 
At another hand, the solid mathematical treatment that can be found in pure mathematical literature  does not usually deal  with the specific tensors encountered in  applications. 
The goal of the current paper is to fill the indicated gap. Although we are studying here only the simplest case of the third-order tensor in three dimensional space, the basic algebraic problems of  reducibility of the unique decomposition and of  non-uniqueness of the irreducible decomposition is exhibited  implicitly. In the case of  physically meaningful third-order tensors with specific prescribed symmetries, these problems do  completely disappear:  The irreducible decomposition of the partially symmetric (or skew-symmetric) tensors is unique. It seems to be a good basis to guess that also the physically meaningful higher-order tensors must have additional symmetries and accept unique irreducible decomposition. This requirement may be a useful guiding principle in  much more complicated situations in general relativity where fourth- and sixth-order tensors  are naturally emerge, see \cite{Baekler:2014kha}, \cite{Itin:2018dru}, \cite{Itin:2016nxk}. 

Recently, the third-order tensors and their decomposition were  studied intensively. Such procedure is usually  based on the {\it harmonic decomposition}  from representation theory, see e.g. \cite{Jerphagnon}, \cite{Auffray08}, \cite{Auffray13}, \cite{Olive}, \cite{Qi}, \cite{Vannucci}  and the references given therein. In this approach, a tensor is decomposed  into rotational invariant sub-tensors, i.e., relative to the group $SO(3,\mathbb R)$.  In the current paper, we apply an alternative strategy, that can be traced to Weyl \cite{Weyl} and {Schouten} \cite{Schouten} . In particular, different types of  decomposition are constructed in a correspondence with the geometric structure on the basic vector space $V$. As a result  different types of decomposition are  encountered  relative to different groups of  transformations. This way we are able to identify the sets of sub-tensors of similar algebraic and physical properties. This procedure is known to be relevant for fourth-order constitutive tensors in electromagnetism \cite{HO}, linear elasticity \cite{IH-JMP}, \cite{Itin-quad}, \cite{YI-MMS20}, and gravity \cite{HK}.  

The organization of the current paper is as follows: 

In Sect. 2, we present the basic facts about the  definition of tensors and relevant groups of their transformations. 
Third-order tensor is defined as a multi-linear map of a Cartesian product of three copies of a vector spaces over a number field into the same field. In this paper, we restrict to the real numbers field $\mathbb R$ and to a $3-$dimensional vector space $V$. The corresponded tensor is denoted by $T^{ijk}$  with  indices changing in the range  $i,j,k=1,2,3$. In Sect. 3, we formulate a list of requirements for {\it irreducible  decomposition of a tensor}. These conditions are presented in two equivalent forms -- in term of sub-tensors and in term of subspaces of the tensor space. 
Our main results are derived in Sect. 4. Step-by-step,  we provide the invariant decomposition of the tensor with respect to the group $GL(3,\mathbb R)$ and its subgroups. The first level of the decomposition is derived by the tools of the symmetry group $S_3$, i.e., Young tableau. Already on this stage, we meet the main problem: The irreducible decomposition is not unique in general. We construct different isotopic types of irreducible decomposition.  Additional levels of decomposition are derived by use of traces (contractions) taken relatively to the metric tensor or permutation pseudo-tensor. In Sect. 5, we apply the general scheme to two most known classes of restricted 3-rd order tensors: the piezoelectric tensor and the  Hall-effect tensor. In Conclusion section, we briefly discuss our results and propose some possible applications.

Notations: We denote third- and second-order tensors by capital Latin letters and strictly distinguish between the upper and lower indices. The  indices take the values  from the range  $i,j,k,\cdots=1,2,3$.  Einstein's summation rule for a pair of indices is assumed everywhere. The symmetrization of a tensor is denoted by $T^{(ij)}=(1/2)(T^{ij}+T^{ji})$, while the antisymmetrization is presented as $T^{[ij]}=(1/2)(T^{ij}-T^{ji})$. The vector spaces corresponding to  specific tensors (tensor spaces) are  denoted by the same letter in the bold font. 

\section{Third-order tensors and their transformation groups}
In this section, we briefly recall some definitions and  notations relevant to tensor algebra. We also discuss  transformation groups  related to various   geometric structures defined on the basic vector space and their relation to 3 dimension tensors. 
\subsection{Third-order tensors}
Let $V$  denotes a 3-dimensional vector space over the real number field $\mathbb R$ with the dual vector space  $V^*$. Most applications in solid-state physics deal with the spaces $V$ and $V^*$ that are isomorphic to the ordinary 3-dimensional space, $V\cong V^*\cong \mathbb R^3$.  Denote a basis of $V$ as $e_i$ and the dual basis of $V^*$ as $\vt^i$. 
The duality of these two bases is defined  by operator relations, as follows 
\begin{equation}
    e_i(\vt^j)=\d^j_i\qquad \mathrm{and}\qquad \vt^i(e_j)=\d^i_j\,,
\end{equation}
where $\d^i_j$ denotes Kronecker's tensor (the set of components of the unit matrix). 
In these bases, the representations of a vector $v\in V$ and  a functional (1-form) $\o\in V^*$ are given by 
\begin{equation}\label{vect-form}
    v=v^ie_i,\qquad \mathrm{and}\qquad\o=\o_i\vt^i\,.
\end{equation}
Here and in the sequel, we use  Einstein's summation rule for two same-named indices in the upper and down positions. In 3-dimensional space, the indices change in the range $i,j,\cdots=1,2,3$.  
The numerical values of the components $v^i$  and $\o_i$ are defined module the transformations of the basis. 
Let the basis $e_i$ of $V$ be transformed into an arbitrary  new basis $e_{i'}$, then the basis $\vt^i$ of $V^*$ is transformed  into the unique new basis $\vt^{i'}$ due to the standard  tensor  law:
\begin{equation}\label{basis-trans}
    e_i \to e_{i'}=R^i_{i'}e_i \qquad \mathrm{and}\qquad
    \vt^i \to \vt^{i'}=R^{i'}_i\vt^i\,. 
\end{equation}
The transformation matrices $R^i_{i'}$ and $R^{i'}_i$ are assumed to be  inverse one to another. In tensor form, it is expressed as 
\begin{equation}
    R^i_{i'}R^{i'}_j=\d^i_j,\qquad\mathrm{and}\qquad R^i_{i'}R^{j'}_i=\d^{j'}_{i'}. 
\end{equation}
Under the basis transformations (\ref{basis-trans}), the components of the  vector $v$ and  the $1-$form $\o$ are transformed respectively as 
\begin{equation}
    v^i\to v^{i'}=R^{i'}_iv^i\,,\qquad 
    \mathrm{and}\qquad \o_i\to\o_{i'}=R^i_{i'}\o_i\,.
\end{equation}
So in (\ref{vect-form}), the vector $v$ and the form $\o$ themselves  are  invariant under arbitrary linear transformations of the basis. 

Tensors are defined as multilinear map from the Cartesian product of  vector spaces into the field $\mathbb R$.  
In particular, the {\it covariant 3-rd order tensor } is defined as 
\begin{equation}
    T:V^*\times V^*\times V^*\to\mathbb R\,.
\end{equation}
With respect to the basis $e_i$, the tensor $T$ is represented by a set of 27 real components  $T^{ijk}$ such that  
\begin{equation}
    T=T^{ijk}e_i\otimes e_j\otimes e_k\,,
\end{equation}
where the tensor product notation $\otimes$  is used.
Under the basis transformations (\ref{basis-trans}), the components of the tensor are transformed as 
\begin{equation}
    T^{ijk}\to T^{i'j'k'}=
    R^{i'}_{i}R^{j'}_{j}R^{k'}_{k}T^{ijk}\,.
\end{equation}
Similarly, the {\it contravariant 3-rd order tensor } is defined as a multilinear map
\begin{equation}
    T:V\times V\times V\to\mathbb R\,.
\end{equation}
Relatively to the basis $\vt^i$ of the dual space $V^*$, this tensor is represented by a set of 27 real components $T_{ijk}$,  as well,  
\begin{equation}
    T=T_{ijk}\vt^i\otimes \vt^j\otimes \vt^k\,.
\end{equation}
Under the linear transformation (\ref{basis-trans}) of the basis, the components of the contravariant tensor are transformed as
\begin{equation}
    T_{ijk}\to T_{i'j'k'}=
    R_{i'}^{i}R_{j'}^{j}R_{k'}^{k}T_{ijk}\,.
\end{equation}
Another class of so-called {\it mixed-type tensors} is defined on the  Cartesian  product of the vector spaces $V$ and $V^*$. For instance, the linear map 
\begin{equation}
    T:V^*\times V\times V\to\mathbb R
\end{equation}
introduces a tensor 
\begin{equation}
    T=T^i{}_{jk}e_i\otimes \vt^j\otimes \vt^k\,
\end{equation}
with a ``mixed" transformation law of the components
\begin{equation}
    T^i{}_{jk}\to T^{i'}{}_{j'k'}=
R^{i'}_{i}R_{j'}^{j}R_{k'}^{k}T^i{}_{jk}\,.
\end{equation}
Analogously, one can define additional mixed-type tensors, such as $T^i{}_j{}^k$ and so on. 
\subsection{Groups relevant for   decomposition of tensors}
 Since the Cartesian product of vector spaces  is not  abelian, the group of permutations (the symmetry group) acted on the set of indices is relevant for the algebra of tensors. 
In index notations, this group, $S_n$,  acts on the components of  a $n$-th order tensor  by permutation of its indices. 
For the $3-$order tensors, we are dealing with the symmetry (permutations) group $S_3$. It is a finite group of  6 independent elements:
\begin{equation}
    S_3=\{ I, (12), (23), (13), (123), (132)\}
\end{equation}
Here the cycle  notations are used.  The element $I$ denotes the identity permutation The  element $(12)$ means a permutation that exchanges  the first and the second indices.   The elements $(23)$ and  $(13)$ are defined similarly.  The element $(123)$ means that the first index goes into the second, the second into the third, and the third into the first. The element $(132)$ is defined in the same way.

%Since a tensor is defined on a vector space $V$, 
A linear transformation of a basis of $V$ is naturally translated into the transformation of the tensor space. For a 3-dimensional  space $V$, transformations form the group of invertible $3\times 3$ matrices  
\begin{equation}
    GL(3, \mathbb R)=\{(3\times 3)-{\mathrm {matrices} }\,\, G \,\,{\mathrm {with }} \,\,{\mathrm {det}}\,\, G\ne 0\}\,.
\end{equation}

The relation between the action on tensors of these two groups, the permutation group $S_p$ and the  general linear group $GL(n, \mathbb R)$,  is managed by the Schur-Weyl duality theorem. It states that due to commutativity of  the simultaneous action of the groups  $GL(n, \mathbb R)$  and $S_p$, the decomposition of the tensor space relative to the symmetry group $S_p$ is invariant under the action of the general linear group $GL(n, \mathbb R) $. An effective way to derive the $S_p$ decomposition is to  apply  the Young diagram technique. In our case, we are dealing with the groups $GL(3, \mathbb R) $ and $S_3$. 

For the vector space $V$  endowed with an additional geometric structure, such as a metric,  the group of transformations is restricted to a subgroup of $GL(3, \mathbb R) $ that preserves the geometric   structure. Consequently, an additional decomposition of the tensor space is admissible. In this paper, we study the decomposition under the group $GL(3, \mathbb R) $ and its  subgroups:
\begin{itemize}
    \item Orthogonal group $O(3, \mathbb R)$ that preserves the scalar product structure on $V$;
    \item Special linear group $SL(3, \mathbb R)$ that preserves the volume element, i.e., the orientation structure on $V$;
    \item Special orthogonal group $SO(3, \mathbb R)$ that preserves the scalar product together with the  orientation.
\end{itemize}
 
%{\color{red}(UP to here!)}

\section{Irreducible decomposition of tensors}
In a chosen basis, a tensor is presented by a large set of independent components. In order to  characterize the algebraic properties  of a tensor, it is useful to divide this  set of  components  into some smaller subsets, with specific algebraic symmetry for each one of them, as
\begin{equation} \label{dec-cond1}
    T^{i\cdots j} =\sum_{p=1}^n{}^{(p)}T^{i\cdots j}  =^{(1)}\!T^{i\cdots j} +^{(2)}\!T^{i\cdots j} +\cdots+^{(p)}\!T^{i\cdots j} +\cdots+^{(n)}\!T^{i\cdots j}\,.
\end{equation}
Such a procedure is called  {\it decomposition  of a tensor}. In order to get an algebraically  meaningful decomposition, we would like   the following conditions to be satisfied:
\begin{itemize}
    \item [(1)] {\it Covariance:} All sub-tensors $^{(p)}\!T^{i\cdots j}$ must be of the same order and of the same shape like the initial tensor $T^{i\cdots j}$ is. In other words, all sub-tensors must have the same number of indices in the same positions.
    \item [(2)] {\it Independence:} The sub-tensors $^{(p)}\!T^{i\cdots j}$ must be linearly independent, i.e.,  any equation of the form  $\sum_{p=1}^n{}\a_p{}^{(p)}T^{i\cdots j}=0$ yields $\a_p=0$ for all $p$.
    \item [(3)] {\it Irreducibility:} The set of sub-tensors $^{(p)}\!T^{i\cdots j}$ must be minimal, i.e., it can not be decomposed successively into some smaller set of sub-tensors.
    \item [(4)] {\it Uniqueness:} There is no an   alternative decomposition, i.e., if the tensor is written as 
    $T^{i\cdots j} =\sum_{p=1}^n{}^{(p)}\check{T}^{i\cdots j}$, in an  addition to (\ref{dec-cond1}), then ${}^{(p)}\check{T}^{i\cdots j}={}^{(p)}{T}^{i\cdots j}$ for all $p$. 
\end{itemize}

Tensors of a specific order form a  vector space of themselves, which can be called a {\it tensor space}. We denote the tensor spaces of the original tensor ${T}^{i\cdots j}$ and of its sub-tensors ${}^{(p)}{T}^{i\cdots j}$ by ${\bf T}$ and ${}^{(p)}\!{\bf T}$, correspondingly. The conditions above can be reformulated in term of tensor subspaces as follows:
\begin{itemize}
    \item [(1)] {\it Covariance:} The subspaces ${}^{(p)}{\bf T}$ are invariant: It means that they are preserved under all prescribed  transformations. 
    \item [(2)]{\it Independence:} The intersections of the subspaces are  trivial: For $p\ne q$, 
    \begin{equation}
    {}^{(p)}{\bf T}\cap{}^{(q)}{\bf T}=\{0\}\,.
\end{equation}
    \item [(3)]{\it Irreducibility:} The subspaces ${}^{(p)}{\bf T}$  are  minimal, i.e., they do not contain any smaller non-zero invariant subspace. 
    \item [(4)] {\it Uniqueness:} The subspaces ${}^{(p)}{\bf T}$ are unique  (up to isomorphism). 
\end{itemize}

Relaying on the conditions above, the decomposition of the tensor (\ref{dec-cond1})  is presented now as a resolution into the {\it direct sum}  of the tensor space ${\bf T}$
\begin{equation} \label{dec-cond2}
    {\bf T} =\bigoplus_{p=1}^n{}^{(p)}{\bf T} =^{(1)}\!{\bf T} \oplus\cdots\oplus^{(p)}\!{\bf T} \oplus\cdots\oplus^{(n)}\!{\bf T}\,.
\end{equation}
In particular, the dimensions of the subspaces satisfy 
\begin{equation}
    {\rm dim}({\bf T})=
    \sum_{p=1}^n{\rm dim}\left({}^{(p)}{\bf T}\right)\,.
\end{equation}

In order to clarify these issues, we present  some simple examples. 
For a general (asymmetric) $2-$nd order covariant  tensor $T^{ij}$, the decomposition into symmetric and skew-symmetric parts 
\begin{equation}
    T^{ij}=S^{ij}+A^{ij}
\end{equation}
where
\begin{equation}
    S^{ij}=T^{(ij)}:=\frac 12 \left(T^{ij}+T^{ji}\right),\qquad 
    A^{ij}=T^{[ij]}:=\frac 12 \left(T^{ij}-T^{ji}\right)
\end{equation}
is unique and irreducible under the action of the group $GL(3,\mathbb R)$. The dimension of the total tensor space  is distributed among the subspaces as $9=6+3$.  

Another example is a mixed-type second order tensor $T^i{}_j$. It is decomposed uniquely and $GL(3,\mathbb R)$-irreducibly into the scalar  and traceless parts, respectively,  
\begin{equation}
    T^i{}_j= {}^{(1)}T^i{}_j+{}^{(2)}T^i{}_j\,,
\end{equation}
where
\begin{equation}
    {}^{(1)}T^i{}_j=\frac 13 T^m{}_m\d^i_j\,,\qquad {}^{(2)}T^i{}_i=0\,. 
    %{}^{(2)}T^i{}_j=T^i{}_j-{}^{(1)}T^i{}_j  \,.
\end{equation}
The dimension of the total space is distributed now as $9=1+8$.  

When  both tensors $T^{ij}$ and $T^i{}_j$  are considered relative to the sub-group $O(3,\mathbb R)$, the decompositions given above turn our to be  reducible. A finer decomposition  can be derived by the use of the metric tensor $g_{ij}={\rm diag}(1,1,1)$. For  $T^{ij}$, it takes the form
\begin{equation}
    T^{ij}={}^{(1)}T^{ij}+{}^{(2)}T^{ij}+{}^{(3)}T^{ij}\,,
\end{equation}
where 
\begin{equation}
    {}^{(1)}T^{ij}=\frac 13g_{mn}T^{mn}g^{ij},\qquad
    {}^{(2)}T^{ij}=T^{(ij)}-{}^{(1)}T^{ij},\qquad 
    {}^{(3)}T^{ij}=T^{[ij]}\,.
\end{equation}
Here the dimension of the total space is  distributed, as $9=1+5+3$,  respectively. This decomposition is unique and $O(3,\mathbb R)$-irreducible.

For  higher order tensors the situation turns out to be much more complicated.  Recall the well-known fact, see e.g.  \cite{Landsberg}:

\noindent 
{\it For a general tensor of order greater than 2, there is no unique irreducible decomposition. In other words, the unique decomposition is reducible, while the irreducible decomposition is not unique.  }

\noindent We demonstrate these properties  explicitly in the next section.

\section{Decomposition of a covariant 3-rd order tensor}
\subsection{$GL(3, \mathbb R)$-decomposition} 
In this section, we consider a bare vector space $V$ without any additional structure. In this case, the  decomposition of the tensor $T^{ijk}$ has to be invariant under  arbitrary invertible transformation of the basis in $V$. 
Due to the Schur-Weyl duality theorem, such $GL(3, \mathbb R)$-decomposition  is equivalent to  the decomposition of $T^{ijk}$ under the permutation group $S_3$.
  \subsubsection{Straightforward decomposition}
In three-dimensional space $V$, a general (non-restricted) third-order tensor $T^{ijk}$ has $3^3=27$ independent components. 
We are looking for an irreducible $GL(3, \mathbb R)$-invariant decomposition of this tensor. Moreover, it is plausible to have a unique decomposition.  
First, we observe that $T^{ijk}$ can be readily decomposed into the sum of three independent $GL(3, \mathbb R)$-invariant parts. 
Indeed, it is enough to define the totally symmetric part
\begin{equation}\label{eq5}
S^{ijk}= T^{(ijk)}:=\frac{1}{3!}\left(T^{ijk}+T^{jki}+T^{kij}+T^{jik}+T^{kji}+T^{ikj}\right)\,,
\end{equation}
and the totally skew-symmetric part
\begin{equation}\label{eq6}
A^{ijk}=T^{[ijk]}:=\frac{1}{3!}\left(T^{ijk}+T^{jki}+T^{kij}-T^{jik}-T^{kji}-T^{ikj}\right)\,.
\end{equation}
Now we extract these two parts from $T^{ijk}$ to obtain the residue part $N^{ijk}$  given by 
\begin{equation}\label{eq7}
    N^{ijk}:=T^{ijk}-S^{ijk}-A^{ijk}=\frac{1}{3}\left(2T^{ijk}-T^{jki}-T^{kij}\right)\,.
\end{equation}
Consequently, we obtain an invariant   decomposition of $T^{ijk}$ into three invariant  parts 
\begin{equation}\label{decomp1}
    T^{ijk}=S^{ijk}+A^{ijk}+N^{ijk}\,.
\end{equation}
The linear independence of these three  sub-tensors  follows from  the symmetry relations 
\begin{equation}\label{N-sym}
    N^{(ijk)}= A^{(ijk)} =0\,,
\qquad 
{\rm and}\qquad 
   N^{[ijk]}=S^{[ijk]}=0\,.
\end{equation}
These symmetry relations also show that the subspaces corresponding to the $S$, $A$, and $N$ tensors are mutually disjoint (up to the zero tensor).

The decomposition (\ref{decomp1}) is invariant under the action of the $GL(3, \mathbb R)$ group. It means that the symmetries of the sub-tensors in (\ref{decomp1}) are preserved under any linear transformation of the basis. In particular, the transformed tensor  $S^{i'j'k'}=R^{i'}_iR^{j'}_jR^{k'}_kS^{ijk}$ is totally symmetric while the transformed tensor $A^{i'j'k'}=R^{i'}_iR^{j'}_jR^{k'}_kA^{ijk}$ is totally skew-symmetric.  Also the residue tensor $N^{ijk}$ preserves its symmetries (\ref{N-sym}) under these  transformations.  Moreover, we observe  that the tensors  $S^{ijk}$, $A^{ijk}$ and $N^{ijk}$ are  defined uniquely. 

Consequently, the tensor space of the third-order tensors is decomposed now into the direct sum of three subspaces 
\begin{equation}
        {\bf T}={\bf S}\oplus{\bf A} \oplus {\bf N}\,.
    \end{equation}
Hence, the dimension of the total space is distributed among these subspaces  as follows
\begin{equation}
    27=10+ 1+ 16\,.
\end{equation}

It is clear, that the sub-tensors  $S^{ijk}$, and $A^{ijk}$ are irreducible. Indeed, any symmetrization (or antisymmetrization)  of the indices in $S^{ijk}$ and $A^{ijk}$ preserves them (up to the total sign) or gives zero. The sub-tensor $N^{ijk}$, however, is reducible. For instance, one can decompose it into the sum of two non-zero parts  $N^{(ij)k}+N^{[ij]k}$. Whether this decomposition is invariant? Is it minimal?
How the invariant and minimal decomposition of $N^{ijk}$ can be derived? We  discuss these issues in the sequel.

\subsubsection{Young's diagrams}
To have an invariant irreducible decomposition of $N^{ijk}$, we apply the machinery of group theory. 
%The well-known Schur–Weyl duality theorem states the correspondence between the decomposition of a tensor space under the general linear group and the symmetry group.
The technical tool used in the symmetry group decomposition is based on Young's diagrams. 
The symmetry group that is relevant to the decomposition of third-order tensors in a  space of an arbitrary dimension is $S_3$. 
For a generic tensor $T_{ijk}$ of 27 independent components, there are three different Young's diagrams depicted as: 
\begin{equation}
    {\Yvcentermath1
 \lambda_1= \yng(3) ~,
\qquad \lambda_2=\yng(1,1,1) ~,
\qquad \lambda_3=\yng(2,1) ~
 }.
\end{equation}
 Each diagram $\lambda$ has a corresponding dimension (the number of standard Young tableaux whose shape is a given Young diagram). This number can be calculated  due to the so-called {\it hook formula} 
\begin{align}
\begin{split}
\text{dim} ~ \lambda &=  \frac{3!}{\prod\limits_{(\a,\b) \in \lambda} \text{hook}(\a,\b)}.
 \label{eq:dimension}
\end{split}
\end{align}
Here, the pair of numbers $(\a,\b)$ denotes the position of a cell in the diagram: $\a$ for the row and $\b$ for the  column. 
For a cell $(\a,\b)$ in the diagram of a shape $\lambda$, the natural number  ``$\text{hook}(\a,\b)$" is defined as the number of boxes that are in the same row to the right of it plus those boxes in the same column below it, plus one (for the box itself). It is called {\it hook length}.  
Using Eq.(\ref{eq:dimension}), we obtain:
\begin{equation}
 \text{dim} ~ \lambda_1 = 1, \qquad   \text{dim} ~ \lambda_2 = 1, \qquad \text{dim} ~ \lambda_3 = 2.
\end{equation}
The corresponding Young's tableaux (the diagrams filling with the numbers 1,2,3) 
are given as 
\begin{equation}\label{tableu}
    \lambda_1 : \quad \Yvcentermath1 \young(123)\qquad\qquad \lambda_2 : \quad \Yvcentermath1 \young(1,2,3)\qquad\qquad \lambda_3 : \quad \Yvcentermath1 \young(12,3)\,, \quad \young(13,2) 
\end{equation}
The rows in these tableaux describe the symmetrization operated on the 
corresponding 
indices, while the columns mean the antisymmetrization. 
Consequently, the identity operator $I$ is decomposed into the sum of three Young operators:
\begin{equation}
    I=P_1+P_2+P_3, 
\end{equation}
where the symmetrization operator $P_1$  and antisymmetrization operator $P_2$ are defined, respectively, as  
\begin{equation}\label{eq.pal}
\begin{aligned}
P_1=I+(12)+(13)+(23)+(123)+(321)\,,\\
P_2=I-(12)-(13)-(23)+(123)+(321)\,.
\end{aligned}
\end{equation}
The operator $P_3$ is represented as the sum of two operators described by the two tableaux $\lambda_3$ given in (\ref{tableu})
\begin{equation}\label{eq2}
P_3=P_{3,1}+P_{3,2}\,.
\end{equation}
Due to the rules mentioned above, the  explicit expressions of these operators are given by 
\begin{equation}\label{eq3}
\begin{aligned}
P_{3,1}=(I+(12))(I-(13))=I+(12)-(13)-(132)\,, \\
P_{3,2}=(I+(13))(I-(12))=I+(13)-(12)-(123)\,.
\end{aligned}
\end{equation}
Here we assumed the rule: the symmetrization operation is applied {\it after } the antisymmetrization one. An alternative order yields two different operators that however equivalent (isotopic) to (\ref{eq3}).  

Hence we obtain the  decomposition of the tensor $T^{ijk}$ into four pieces:
\begin{eqnarray}\label{eq42}
T^{ijk}&=&S^{ijk}+A^{ijk}+N^{ijk}\nonumber\\
&=&S^{ijk}+A^{ijk}+\left(N_1^{ijk}+N_2^{ijk}\right)\,.
\end{eqnarray}
Here $S^{ijk}$ and $A^{ijk}$ are the familiar totally symmetric and totally skew-symmetric parts, (\ref{eq5}) and (\ref{eq6}) respectively. The residue part $N^{ijk}$ given in (\ref{eq7}) is presented now as  a sum of two parts. 
The explicit expressions of these two  sub-tensors are calculated due to the operators (\ref{eq3})  correspondingly to the two latter tableaux in (\ref{tableu}) 
\begin{equation}\label{eq7-0}
N_1^{ijk}=P_{3,1}\left(T^{ijk}\right)=\frac{1}{3 } \left(T^{ijk}+T^{jik}-T^{kji}-T^{kij}\right)
\end{equation}
and
\begin{equation}\label{eq8x}
N_2^{ijk}=P_{3,2}\left(T^{ijk}\right)=\frac{1}{3}
\left(T^{ijk}-T^{jik}+T^{kji}-T^{jki}\right)\,.
\end{equation}
%Observe the symmetries
%\begin{equation}
%    N_1^{ijk}=N_1^{jik},\qquad %N_2^{ijk}=N_2^{kji}\,.
%\end{equation}
It is clear that these two tensors are expressed via $N^{ijk}$ only, i.e., 
\begin{equation}\label{eq7X}
N_1^{ijk}=\frac{1}{3 } \left(N^{ijk}+N^{jik}-N^{kji}-N^{kij}\right)
\end{equation}
and
\begin{equation}\label{eq8X}
N_2^{ijk}=\frac{1}{3}
\left(N^{ijk}-N^{jik}+N^{kji}-N^{jki}\right)\,.
\end{equation}
Notice the symmetries of these two tensors following immediately from Eqs.(\ref{eq3})
\begin{equation}
    N_1^{ijk}=N_1^{jik}\quad {\rm and}\quad  N_2^{ijk}=N_2^{kji}\,.
\end{equation}
\subsubsection{Alternative decomposition}
We can easily see that the decomposition (\ref{eq42}) is not unique. 
For example, if we assume an opposite order rule: the antisymmetrization  applied after the symmetrization, we would  obtain
\begin{equation}\label{N2-decomp}
N^{ijk}=\tilde{N}_1^{ijk}+\tilde{N}_2^{ijk}\,.
\end{equation}
The corresponding operators are
\begin{equation}\label{eq3xx}
\begin{aligned}
\tilde{P}^1_3=(I+(23))(I-(12))=I+(23)-(12)-(132)\,. \\
\tilde{P}^2_3=(I+(12))(I-(23))=I+(12)-(23)-(123)\,.
\end{aligned}
\end{equation}
Consequently, these two alternative sub-tensors are expressed as
\begin{equation}\label{eq7AA}
\tilde{N}_1^{ijk}=\frac{1}{3 }\left(T^{ijk}-T^{jik}+T^{ikj}-T^{kij}\right)=\frac{1}{3 }\left(N^{ijk}-N^{jik}+N^{ikj}-N^{kij}\right)\,, 
\end{equation}
and
\begin{equation}\label{eq8a}
\tilde{N}_2^{ijk}=\frac{1}{3}\left(T^{ijk}-T^{ikj}+T^{jik}-T^{jki}\right)=\frac{1}{3}\left(N^{ijk}-N^{ikj}+N^{jik}-N^{jki}\right)\,.
\end{equation}
The $\tilde{N}$-decomposition is not the only possible alternative to the $N$-decomposition. Indeed, we can use  various sequences of the permutation operators to produce different  decompositions. 
Particularly, we  use in the sequel a decomposition based on the operators 
\begin{equation}\label{eq3x}
\begin{aligned}
\hat{P}^1_3=(I-(23))(I+(12))= I-(23)+(12)-(132)\,, \\
\hat{P}^2_3=(I-(12))(I+(23))=I-(12)+(23)-(123)\,.
\end{aligned}
\end{equation}
The corresponding subtensors are
\begin{equation}\label{eq7A}
\hat{N}_1^{ijk}=\frac{1}{3}(T^{ijk}-T^{ikj}+T^{jik}-T^{kij})=\frac{1}{3}(N^{ijk}-N^{ikj}+N^{jik}-N^{kij})\,,
\end{equation}
and
\begin{equation}\label{eq8b}
\hat{N}_2^{ijk}=\frac{1}{3}(T^{ijk}-T^{jik}+T^{ikj}-T^{jki})=\frac{1}{3}(N^{ijk}-N^{jik}+N^{ikj}-N^{jki})\,.
\end{equation}
 All these {\it isotopic decompositions} of the tensor $N^{ijk}$ are completely equivalent. All of them are irreducible and invariant under the $GL(3,{\mathbb R})$ transformations of the basis.  
For higher-order tensors, the non-uniqueness of the irreducible decomposition is a well-established fact, see e.g.  \cite{Landsberg}.
\subsubsection{Tensor space decomposition}
 %The set of the 3-rd order tensors $T^{ijk}$ compose a vector space of dimension 27. We denote this {\it tensor space } ${\bf T}$.  
 With respect to the  decomposition (\ref{eq42}),   the tensor space ${\bf T}$ is decomposed into the direct sum of the subspaces  ${}^{(p)}\!{\bf T}$ that  correspond to the tableaux $\lambda_p$ given in (\ref{tableu}). 
The dimensions of these subspaces 
can be calculated  by the following combinatorics   formula, see. e.g. \cite{Cvitanovic}, 
    \begin{align}
\begin{split}
\text{dim} ~ {}^{(p)}\!{\bf T} &= \prod\limits_{(\alpha,\beta) \in \lambda_p} \frac{3 + \beta - \alpha}{\text{hook}(\alpha,\beta)}  . \label{eq:dimension1}
\end{split}
\end{align}
 Recall that the pair of integers $(\alpha,\beta)$ numerates a cell of the diagram standing in the $\alpha$'s raw and $\beta$'s column. The summation is provided for all cells of the corresponding diagram. 
Accordingly,  we calculate
\begin{equation}
\text{dim}\,{\bf S} = 10,\qquad 
\text{dim}\,{\bf A} =  1 , \qquad
\text{dim}\,{\bf N}_1=\text{dim}\,{\bf N}_2=8\,.
\end{equation}
Using the consideration above we derive the geometrical meaning of the  decomposition (\ref{tableu}):

\begin{prop}
Let a third-order covariant tensor $T^{ijk}$ and the corresponding tensor space ${\bf T}$ be given. 
\begin{itemize}
    \item The tensor space $\bf T$ is decomposed  into the direct sum of three subspaces 
    \begin{equation}\label{dec1}
        {\bf T}={\bf S}\oplus{\bf A} \oplus {\bf N}\,.
    \end{equation}
    This decomposition is uniquely but reducible. 
 \item  The subspace ${\bf N}$ is   decomposed additionally into the direct sum of two smaller subspaces 
    \begin{equation}\label{dec2}
       {\bf N}= {\bf N}_1\oplus{\bf N}_2\,.
    \end{equation}
    This decomposition is irreducible but not unique. There are isotopic subspaces ${\bf N}_1, {\tilde{\bf N}}_1, {\check{\bf N}}_1$ and so on. 
    \item 
The dimension of the total tensor space is distributed among the subspaces accordingly to:
\begin{equation}\label{dec3}
    3 \times 3 \times 3  = 10 + 1 + (8+ 8)\,.
\end{equation}
\item The subspases ${\bf S}$, ${\bf A}$,   and ${\bf N}_1$, ${\bf N}_2$ are invariant  under $GL(3,\mathbb R)$ transformations of the basis.
\end{itemize}
\end{prop}

\subsection{$O(3,\mathbb R)$-decomposition}
\subsubsection{Metric tensor}
In this section, we consider the vector space $V$  endowed with the scalar product structure. For arbitrary pair of vectors $x, y\in V$, the scalar product is defined by the use of an invertible matrix $g_{ij}$.
In a chosen basis $e_i$, the vectors are represented as  $x=x^ie_i$ and $y=y^ie_i$ while the scalar product is given by 
\begin{equation}
    (x,y)=g_{ij}x^iy^j\,.
\end{equation}
 In a transformed basis $e_{i'}=R_{i'}^ie_i$, it is expressed as 
\begin{equation}
    (x,y)=g_{i'j'}x^{i'}y^{j'}=
\left(g_{i'j'}R^{i'}_iR^{j'}_j\right)x^{i}y^{j}\,.
\end{equation}
Thus the matrix $g_{ij}$ is a tensor with the transformation law
\begin{equation}
    g_{ij}\to g_{i'j'}=R^i_{i'}R^j_{j'}g_{ij}\,.
\end{equation}
In Euclidean space, the metric tensor has a fixed diagonal form $g_{ij}={\rm diag}(1,1,1)$. Such a metric is preserved under the transformations of basis restricted to a subgroup of $GL(3,{\mathbb R})$ which elements  satisfy the requirement $RR^T=I$. This is the orthogonal group $O(3,{\mathbb R})$. 
\subsubsection{Trace vectors in symmetric and skew-symmetric parts}
We are looking now for a decomposition of the tensor $T^{ijk}$ under the orthogonal group. Using the metric tensor $g_{ij}$, we can construct from $T^{ijk}$ various tensors of smaller orders.  Here we refer to  the sum of tensor components in some two fixed indices, one upper and one lower, as  a {\it trace of a tensor}. 
The traces are tensors by themselves of the rank $p-2$, where $p$ is the rank of the original tensor. 
From the tensor $T^{ijk}$, we can extract three traces as follows 
\begin{equation}\label{3vectors}
 u^k=g_{ij}T^{ijk}\,,\qquad   v^k=g_{ij}T^{ikj}\,,\qquad  w^k=g_{ij}T^{kij}\,.
\end{equation}
They are transformed as vectors under a  transformation from $O(3,{\mathbb R})$. 
For a generic tensor $T^{ijk}$, these three  vectors are linearly independent. Otherwise, we would have a linear relation between the components of  $T^{ijk}$, then  the tensor would not be general. We  discuss such special  tensors in the  sequel. 

In  addition to the decomposition (\ref{eq42}), we are  able  now to derive a successive decomposition of 
the tensor $T^{ijk}$ by using 
the trace vectors  $u^k,v^k,w^k$. 
Let us  start with the symmetric part (\ref{eq5}).  Due to the total symmetry of  $S^{ijk}$, we have  $g_{ij}S^{ijk}=g_{ij}S^{ikj}=g_{ij}S^{kij}$. Thus,  there is only one independent  trace vector that can be constructed from the symmetric part 
\begin{equation}
   \alpha^k= g_{ij}S^{ijk}=\frac{1}{3}(u^k+v^k+w^k)\,.
\end{equation}
Now the tensor $S^{ijk}$ can be decomposed into the sum of two independent symmetric sub-tensors: 
\begin{equation}\label{S-decomp0}
    S^{ijk}=K^{ijk}+R^{ijk}\,,
\end{equation}
where the {\it trace part } is defined as 
    \begin{equation}\label{S-decomp1}
    K^{ijk}=\frac 15\left(\alpha^i g^{jk} +\alpha^j g^{ik}+ \alpha^k g^{ij}\right)\,,
\end{equation}
while $R^{ijk}=S^{ijk}-K^{ijk}$ is the {\it totally traceless part.} The factor $(1/5)$ in (\ref{S-decomp1}) follows from the traceless relations 
\begin{equation}
\label{trace}
  g_{ij}R^{ijk}=0, \quad  g_{ik}R^{ijk}=0,   \quad g_{jk}R^{ijk}=0 \,.
\end{equation}
In term of the corresponding tensor spaces, Eq.(\ref{S-decomp0})  means the decomposition of the spase $\bf{S}$ into a direct sum of two invariant subspaces
\begin{equation}\label{S-decomp3}
    \bf{S}=\bf{K}\oplus\bf{R}
\end{equation}
with the dimension reduction
\begin{equation}\label{S-decomp3x}
    10=3+ 7\,.
\end{equation}

The second  part $A^{ijk}$ of the tensor $T^{ijk}$ is totally skew-symmetric. All its traces are zero so it does not have any invariant subspaces. It is trivial since the corresponding space $\bf{A}$ is 1-dimensional. 

\subsubsection{Trace vectors and decomposition of the mixed-symmetry part}
We consider now the third part $N^{ijk}$.  It has two invariant subspaces $N_1^{ijk}$ and $N_2^{ijk}$, that are useful to consider separately. 
Due to the symmetry $N_1^{ijk}=N_1^{jik}$,  there are  two different  trace  vectors 
\begin{equation}
%\beta^k=
g_{ij}N_1^{ijk}=\frac{2}{3}(u^k-w^k)\,,
\qquad
{\rm and}\qquad  
%\gamma^k=
g_{ij}N_1^{ikj}=\frac{1}{3}(w^k-u^k)\,.
\end{equation}
These vectors are linearly dependent,  %$\gamma^k=-(1/2)\beta^k$
so we are left with  only one independent vector 
\begin{equation}
    \beta^k=g_{ij}N_1^{ijk}=\frac{2}{3}(u^k-w^k)
\end{equation}
%$\beta^k$
only.  
Now the tensor $N_1^{ijk}$ can be decomposed into the sum of two independent parts -- the trace and the traceless one
\begin{equation}
    N_1^{ijk}=M_1^{ijk}+P_1^{ijk}\,. 
\end{equation} 
Due to the symmetry of $N_1^{ijk}$ , the general form of the trace part can be  considered as 
\begin{equation}\label{m11}
M_1^{ijk}=x(\beta^i g^{jk}+\beta^j g^{ik})+y\beta^k g^{ij}\,.
\end{equation}
We require the  residue part $P_1^{ijk}$ to be totally traceless, i.e., 
\begin{equation}\label{trace11}
  g_{ij}P_1^{ijk}=0,\qquad g_{ik}P_1^{ijk}=0,\qquad g_{jk}P_1^{ijk}=0\,.
\end{equation}
From these linear equations, it follows that $x=-(1/4)$ and $y=1/2$.
%\begin{equation}
 %   a=-\frac{1}{4},\quad b=\frac{1}{2}.
%\end{equation}
Hence, we derived two parts of the tensor $N_1^{ijk}$. Explicitly, 
\begin{equation}\label{m11x}
M_1^{ijk}=\frac{1}{4}\left(2\beta^k g^{ij}-\beta^i g^{jk}-\beta^j g^{ik}\right)\,,
\end{equation}
and 
\begin{equation}\label{m11xx}
P_1^{ijk}=N_1^{ijk}-M_1^{ijk}\,.
\end{equation}
Accordingly, we have a direct sum decomposition of the  subspace ${\bf N}_1$ into the sum of two  subspaces with the  corresponding dimensions 
\begin{equation}
    {\bf N}_1={\bf M}_1\oplus {\bf P}_1\,, \qquad\qquad 8=3\oplus 5\,.
\end{equation}

 For the tensor $N_2^{ijk}$ with the symmetry $N_2^{ijk}=N_2^{kji}$, we can  define two trace vectors 
 \begin{equation}
%\mu^k=
g_{ij}N_2^{ikj}=\frac{2}{3}(v^k-w^k)\,,
%\end{equation}
\qquad {\rm and}\qquad
%\begin{equation}
%\lambda^k=
g_{ij}N_2^{ijk}=\frac{1}{3}(w^k-v^k)\,,
\end{equation}
which turn out to be linearly dependent.  %$\lambda^k=-(1/2)\mu^k$. 
We choose an independent vector
\begin{equation}
    \gamma^k=
g_{ij}N_2^{ikj}=\frac{2}{3}(v^k-w^k)\,.
\end{equation}
Now, $N_2^{ijk}$ can be decomposed into the sum of two independent parts:
\begin{equation}
    N_2^{ijk}=M_2^{ijk}+P_2^{ijk}\,,
\end{equation}
where, similarly to $M_1^{ijk}$, 
we define  the trace part 
\begin{equation}
\label{m22}
 M_2^{ijk}=\frac{1}{4}\left(2\gamma^j g^{ik}-\gamma^i g^{jk}-\gamma^k g^{ij}\right)\,,
\end{equation}
while $P_2^{ijk}=N_2^{ijk}-M_2^{ijk}$ is a traceless tensor, i.e., 
\begin{equation}
\label{trace12}
  g_{ij}P_2^{ijk}=0,\qquad g_{ik}P_2^{ijk}=0,\qquad g_{jk}P_2^{ijk}=0\,.
\end{equation}
Consequently, we have the direct sum decomposition of the  subspace ${\bf N}_2$ into the sum of two  subspaces with the  corresponding dimensions 
\begin{equation}
    {\bf N}_2={\bf M}_2\oplus {\bf P}_2\,, \qquad\qquad 8=3\oplus 5
\end{equation}

Recall that the decomposition of the tensor $N^{ijk}$ into two parts $N_1^{ijk}$ and $N_2^{ijk}$ is not unique but irreducible.  Its trace decomposition, however, is unique but reducible. Indeed, it is enough to define 
\begin{equation}\label{M+p-tensor}
    N^{ijk}=M^{ijk}+P^{ijk}\,,
\end{equation}
where the tensor $M^{ijk}=M_1^{ijk}+M_2^{ijk}$ reads
\begin{equation}\label{totalM}
    M^{ijk}=\frac{1}{4}\left((2\beta^k-\gamma^k) g^{ij}-(\beta^i+\gamma^i) g^{jk}+ (2\gamma^j-\beta^j) g^{ik}\right)\,,
\end{equation}
or, equivalently,
\begin{equation}\label{M-tensor}
    M^{ijk}=\frac{1}{6}\left((2u^k-v^k-w^k) g^{ij}+(2w^i-u^i-v^i) g^{jk}+ (2v^j-u^j-w^j) g^{ik}\right)\,.
\end{equation}
This tensor has 6 independent components (the number of  components of two independent vectors).
Since the tensors $M^{ijk}$ and $P^{ijk}$ are  independent, the second traceless tensor 
\begin{equation}
    P^{ijk}=P_1^{ijk}+P_2^{ijk}
\end{equation}
is left with 10 independent components.  

%\subsubsection{Subspace description}
Accordingly, the metric structure on the vector space $V$ provides a finer decomposition of the tensor $T^{ijk}$. In geometrical description, our results can be formulated as follows:
\begin{prop} Let a general tensor $T^{ijk}$ with the corresponding tensor space ${\bf T}$ be given.
\begin{itemize}
    \item 
 The vector space of the tensor $T^{ijk}$ is decomposed uniquely into the sum of 5  independent invariant subspaces
\begin{eqnarray}\label{2dec1}
    {\bf T}&=&{\bf S}\oplus{\bf A}\oplus{\bf N} \nonumber\\
    &=&({\bf K}\oplus{\bf R})\oplus{\bf A}\oplus({\bf M}\oplus{\bf P})\,.
\end{eqnarray}
\item The dimension of the total space is distributed between these subspaces as follows
\begin{equation}\label{dec3a}
   27  = (3 + 7)+ 1 + (6+ 10)\,.
\end{equation}
\item The subspaces ${\bf M}$ and ${\bf P}$ are decomposed irreducible but not uniquely into the sum of two subspaces with the corresponding dimensions 
\begin{equation}\label{2dec2}
    {\bf M}={\bf M}_1\oplus {\bf M}_2\qquad 6=3+ 3
\end{equation}
and 
\begin{equation}\label{2dec3}
    {\bf P}={\bf P}_1\oplus {\bf P}_2\qquad 10=5 + 5\,.
\end{equation}
\item The decomposition is invariant under the action of the group $O(3,\mathbb R)$.
\end{itemize}
\end{prop}
\subsubsection{Orthogonality of irreducible parts}
With a metric tensor at hand we can define a  scalar product of two  tensors of the same index structure. For two  third-order tensors $A^{ijk}$ and $B^{ijk}$, it is defined as 
\begin{equation}\label{ort1x}
    (A,B):=A^{ijk}B^{mnp}g_{im}g_{jn}g_{kp}:=A^{ijk}B_{ijk}\,.
\end{equation}
In particular, the scalar square of the tensor $T^{ijk}$ is defined as 
\begin{equation}\label{ort1}
    (T,T):=T^{ijk}T^{mnp}g_{im}g_{jn}g_{kp}=T^{ijk}T_{ijk}\,.
\end{equation}
Since it is non-negative, the Euclidean norm of a tensor  is given in the standard form  $||T||=(T,T)^{1/2}$. 
In Eqs.(\ref{ort1x},\ref{ort1}), we use  the contravariant tensor
\begin{equation}\label{ort2}
T_{ijk}:=g_{im}g_{jn}g_{kp}T^{mnp}\,.
\end{equation}
Notice that for the Euclidean metric $g_{ij}={\rm diag}(1,1,1)$, the tensors $T_{ijk}$ and $T^{ijk}$ have the same numerical components. So they can be referred to as different versions of the same tensor. This assumption is widely used in literature, even without notice that it is true only for Cartesian tensors with the Euclidean metric. 

%Due to Eq.(\ref{ort1}),  the Euclidean norm of the tensor $T^{ijk}$is defined in the standard way $||T||=(T,T)^{1/2}$. It is positive defined

Let us observe the following properties of the scalar product of tensors that are useful in manipulation with the indices:
\begin{itemize}
    \item For an arbitrary permutation of indices  $\sigma$,
    \begin{equation}\label{perm1}
        (\sigma A,\sigma B)=(A,B)\,.
    \end{equation}
    \item For an arbitrary permutation of indices  $\sigma$ and its inverse $\sigma^{-1}$, 
    \begin{equation}\label{perm2}
        (\sigma A, B)=(A,\sigma^{-1}B)\,.
    \end{equation}
    \item For a linear combination   $P=\sum \a_i\sigma_i$ of  permutations $\sigma_i$ with arbitrary real coefficients $\a_i$,
    \begin{equation}\label{perm3}
        (PA,B)=(\sum \a_i\sigma_iA,B)=(A,\sum \a_i\sigma_i^{-1}B)\,.
    \end{equation}
    In particular, using (\ref{eq.pal}) for the fully symmetric and  skew-symmetric tensors, we have respectively
    \begin{equation}\label{sym-pr3a}
        A^{(ijk)}B_{ijk}=A^{ijk}B_{(ijk)}= A^{(ijk)}B_{(ijk)}
    \end{equation}
    and
    \begin{equation}\label{sym-pr3b}
        A^{[ijk]}B_{ijk}=A^{ijk}B_{[ijk]}= A^{[ijk]}B_{[ijk]}\,.
    \end{equation}
\end{itemize}
Now we are able to formulate the following statement:
\begin{prop}
Let the tensor $T^{ijk}$ be uniquely decomposed as
\begin{eqnarray}\label{pr3a}
     T^{ijk}=S^{ijk}+A^{ijk}+N^{ijk}%\\
    % &=&(K^{ijk}+R^{ijk})+A^{ijk}+(M^{ijk}+P^{ijk})\,.
\end{eqnarray}\label{ort3}
where 
\begin{equation}\label{pr3b}
    S^{ijk}=K^{ijk}+R^{ijk}\,,\qquad N^{ijk}=M^{ijk}+P^{ijk}\,.
\end{equation}
Then all five sub-tensors given above (and their  corresponding subspaces) are mutually orthogonal one to another. 
\end{prop}
\begin{proof}
We provide a proof for two sequel levels of the decomposition. %Observe a general fact: the contravariant version of a tensor has the same symmetries as the covariant one. 
First, we prove that the tensors in (\ref{pr3a})
%\begin{equation}\label{ort3}
%    T^{ijk}=S^{ijk}+A^{ijk}+N^{ijk}\,.
%\end{equation}
are mutually orthogonal.
The relations
\begin{equation}\label{ort4}
 S^{ijk}A_{ijk}= S^{ijk}N_{ijk} =A^{ijk}N_{ijk} =0
\end{equation}
follow immediately from Eqs.(\ref{sym-pr3a},\ref{sym-pr3b}). 
For instance,
\begin{equation}
    S^{ijk}N_{ijk}= S^{(ijk)}N_{ijk}= S^{(ijk)}N_{(ijk)}=0\,.
\end{equation}
Thus the subspaces $\bf{S},\bf{A},\bf{N}$ are mutually orthogonal one to another. 
On the second level, let us prove  that the tensors $K^{ijk}$ and $R^{ijk}$ are orthogonal. 
Since the tensor $R_{ijk}$ is traceless, 
\begin{equation}
    K^{ijk}R_{ijk}=\frac 15\left(\alpha^i g^{jk} +\alpha^j g^{ik}+ \alpha^k g^{ij}\right)R_{ijk}=0
\end{equation}
Similarly the tensors $M^{ijk}$ and $P^{ijk}$ are orthogonal one to another. 
Thus we conclude that all 5 independent parts of the tensor $T^{ijk}$ are mutually  orthogonal.
\end{proof}
The behaviour of the (non-unique) irreducible parts of $T^{ijk}$ is more complicated. Let us check the scalar product of the tensors $N_1^{ijk}$ and $N_2^{ijk}$. We use the operator representation of these  tensors (\ref{eq3}) and apply the fact that the simple cycle  permutations $(12), (13)$ and $(23)$ are self inverse. Due to (\ref{perm2}), 
\begin{eqnarray}
    N_1^{ijk}N_{2ijk}&=&
    (I+(12))(I-(13))T^{ijk}N_{2ijk}\nonumber\\
    &=&
    (I-(13))T^{ijk}(I+(12))N_{2ijk}\nonumber\\
    &=&T^{ijk}(I-(13))(I+(12))N_{2ijk}\nonumber\\
    &=&
    T^{ijk}(I-(13))(I+(12))(I+(13))(I-(12))T_{ijk}\ne 0\,.
\end{eqnarray}
The product of the operators in the middle is not zero, thus the tensors $N_1^{ijk}$ and $N_2^{ijk}$ are not orthogonal.

We can, however, identify two orthogonal proper subspaces of the space ${\bf{N}}$. Since $(I+(12))(I-(12))=0$, we have 
\begin{eqnarray}
    N_1^{ijk}{\hat N}_2^{ijk}&=&
    (I+(12))(I-(13))T^{ijk}N_2^{ijk}\nonumber\\
    &=&
    (I-(13))T^{ijk}(I+(12))N_2^{ijk}\nonumber\\
    &=&T^{ijk}(I-(13))(I+(12))N_2^{ijk}\nonumber\\
    &=&
    T^{ijk}(I-(13))(I+(12))(I-(12))(I+(23))T^{ijk}=0\,.
\end{eqnarray}
Notice, however, that we do not have here the direct sum decomposition, 
\begin{equation}
    {\bf N}\ne  {\bf N}_1\oplus {\bf{\hat N}}_2\,.
\end{equation}

\subsection{$SL(3,\mathbb R)$-decomposition} \subsubsection{Permutation tensor}
Consider now a 3-dimensional vector space $V$  endowed with the volume element instead of the scalar product. 
The volume element structure is uniquely determined by the permutation  pseudo-tensor $\epsilon_{ijk}$ of Levi-Civita. It is defined as 
\begin{equation}\label{ep-def}
    \epsilon_{ijk}=\left\{ \begin{array}{rl}
         1 & \mbox{if $(ijk)$ is an even permutation of (123)};\\
        -1 & \mbox{if $(ijk)$ is an odd permutation of (123)};\\
        0&\mbox{otherwise} 
        .\end{array} \right.
\end{equation}
Notice the contraction relations 
\begin{equation}\label{epsilon1}
    \epsilon_{ijk}\epsilon^{mnk}=\d^m_i\d^n_j-\d^m_j\d^n_i\,,
\end{equation}
and 
\begin{equation}\label{epsilon2}
    \epsilon_{ijk}\epsilon^{mjk}=2\d^m_i\,,\qquad \epsilon_{ijk}\epsilon^{ijk}=6\,
\end{equation} 
that are useful for manipulations with the tensor $\epsilon_{ijk}$. The upper indicied tensor $\epsilon^{ijk}$ is defined with the same numerical values as in (\ref{ep-def}). \footnote{This definition is in a correspondence with the Euclidean signature. For Lorentzian signature, different sign assumptions for  $\epsilon^{ijk}$ and $\epsilon_{ijk}$ must be involved.}

Under a linear transformation $e^i\to e^{i'}=R^{i'}_ie^i$ of a basis in $V$, 
  Levi-Civita's pseudo-tensor is transformed as
\begin{equation}
    \epsilon_{ijk}=({\rm det\,}R)\,R^{i'}_iR^{j'}_jR^{k'}_k
    \epsilon_{i'j'k'}\,,
\end{equation}
where ${\rm det\,}R$ denotes  the determinant of the matrix $R^{i'}_i$. 
Consequently, the numerical values of $\epsilon_{ijk}$ are preserved only under transformations with matrices of unit determinant, i.e.,     
${\rm det\,}R=1$. This condition  defines a subgroup of $GL(3,\mathbb R)$ -- the {\it special linear group} $SL(3,\mathbb R)$.

\subsubsection{Permutation tensor and symmetric and skew-symmetric parts}
Now we try to apply the tensor $\epsilon_{ijk}$ 
to a decomposition of a generic third-order tensor $T^{ijk}$. The full contraction of these two tensors produces a scalar 
\begin{equation}\label{A-def}
   A=\frac 16\epsilon_{ijk} T^{ijk}\,.
\end{equation}
Notice that $A$ is a pseudo-tensor under transformations from $GL(3,\mathbb R)$. Under transformations from $SL(3,\mathbb R)$, it is a proper tensor.

Due to the symmetries of the parts $S^{ijk}$ and $N^{ijk}$, their full contraction with $\epsilon_{ijk}$ vanish
\begin{equation}
   \epsilon_{ijk} S^{ijk}=0\,,\qquad \epsilon_{ijk} N^{ijk}=0\,.
\end{equation}
Consequently, only the totally antisymmetric part of $T^{ijk}$ contributes into the lhs of Eq.(\ref{A-def}), i.e.,
\begin{equation}
   A=\frac 16\epsilon_{ijk} A^{ijk}\,.
\end{equation}
Accordingly to (\ref{epsilon2}),  the inverse relation  is given by
\begin{equation}\label{Aep}
    A^{ijk}= A\epsilon^{ijk}\,.
\end{equation}
Thus the totally antisymmetric irreducible part $A^{ijk}$ of the tensor $T^{ijk}$ is completely  expressed by the scalar $A$.  

\subsubsection{Pseudo-tensors}
We define now  partial contractions of the tensors $T^{ijk}$ and $\epsilon_{ijk}$ with two indices summed. 
For a generic tensor $T^{ijk}$, there are three different possible contractions of this type
\begin{equation}\label{ABC}
    A_i{}^m=\epsilon_{ijk}T^{mjk}\,,\qquad 
    B_i{}^m=\epsilon_{ijk}T^{kmj}\,,\qquad 
    C_i{}^m=\epsilon_{ijk}T^{jkm}\,.
\end{equation}
Observe that these three tensors have the same trace
\begin{equation}
    A_i{}^i=B_i{}^i=C_i{}^i=6A\,.
\end{equation}
Let us substitute  the  $GL(3\mathbb R)$-decomposition of the tensor, $T^{ijk}=S^{ijk}+A^{ijk}+N^{ijk}$ into Eqs.(\ref{ABC}). Since a   contraction of a symmetric and antisymmetric sets of the same indices  is zero, the totally symmetric part $S^{ijk}$ does not contribute to the 2-indicied tensors (\ref{ABC}).  In order to derive the contributions of  two additional 3-rd order  tensors  $A^{ijk}$ and $N^{ijk}$, it is convenient to calculate their sum 
\begin{eqnarray}
    A_i{}^m+B_i{}^m+C_i{}^m&=&\epsilon_{ijk}\left(T^{mjk}+T^{kmj}+T^{jkm}\right)\nonumber\\&=&
    \frac 12\epsilon_{ijk}\left(T^{mjk}-T^{mkj}+T^{kmj}-T^{jmk}+T^{jkm}-T^{kjm}\right)\nonumber\\&=&
    3\epsilon_{ijk}T^{[mjk]}=3\epsilon_{ijk}A^{mjk}\,.
\end{eqnarray}
Due to (\ref{Aep}) and (\ref{epsilon2}), it means  
 \begin{equation}\label{ABC1x}
     A_i{}^m+B_i{}^m+C_i{}^m=3A\epsilon_{ijk}\epsilon^{mjk}=6A\d^m_i\,.
 \end{equation}
%Accordingly to (\ref{epsilon2}) and (\ref{Aep}), we have for the totally antisymmetric tensor $A^{ijk}$, 
%\begin{equation}
%    \epsilon_{ijk}A^{mjk}=\epsilon_{ijk}A^{jmk}=\epsilon_{ijk}A^{mkj}=2A\d_i^m\,.
%\end{equation}
Define the traceless combinations 
\begin{equation}\label{ABC-check}
    \check{A}_i{}^m:=A_i{}^m-2A\d_i^m\,,\qquad 
    \check{B}_i{}^m:=B_i{}^m-2A\d_i^m\,,\qquad 
    \check{C}_i{}^m:=C_i{}^m-2A\d_i^m\,.
\end{equation}
These tensors are traceless, $\check{A}_i{}^i=\check{B}_i{}^i=\check{C}_i{}^i=0$, so every one of these tensors has 8 independent components. Eq. (\ref{ABC1x}) yields 
 \begin{equation}\label{ABC-check-sum}
    \check{A}_i{}^m+\check{B}_i{}^m+\check{C}_i{}^m=0\,.
 \end{equation}
 It means that only $16$ components are independent. 
In fact, only the mixed symmetry part $N^{ijk}$ of $T^{ijk}$  contributes to the tensors (\ref{ABC-check}). 
Indeed, the contribution of the symmetric tensor $S^{ijk}$ multiplied with the antisymmetric Levi-Civita's  tensor is identically zero, while   the tensor $A^{ijk}$ compensates the scalar parts in (\ref{ABC-check}). 
Consequently, (\ref{ABC-check}) can be rewritten via the  tensor $N^{ijk}$ only
\begin{equation}\label{ABC-check-N}
    \check{A}_i{}^m=\epsilon_{ijk}N^{mjk}\,,\qquad 
    \check{B}_i{}^m=\epsilon_{ijk}N^{kmj}\,,\qquad 
    \check{C}_i{}^m=\epsilon_{ijk}N^{jkm}\,.
\end{equation}

Due to (\ref{ABC-check-sum}) the tensors $\check{A}_i{}^m,\check{B}_i{}^m$ and $\check{C}_i{}^m$ are linearly dependent, so we have  two linearly independent traceless 2nd-order tensors of 8 independent components each.   Every two of the tensors given in  (\ref{ABC-check}),  or every two independent  linear combinations of them, can be chosen as  a representation of the irreducible part $N^{ijk}$ of  16 independent components.

\subsubsection{Pseudotensor representation of the mixed-symmetry part}
In order to demonstrate the full equivalence between the 3-rd order tensor $N^{ijk}$ and the 2-nd order tensors $A_i{}^j,B_i{}^j $, and $C_i{}^j$, we have to solve the linear system (\ref{ABC-check-N}).
It is a system of 24 linear equations in 16 independent variables $N^{ijk}$. Due to 8 constraints given in   (\ref{ABC-check-sum}), this system is well-posed. 
Substituting the decomposition $N^{ijk}=N_1^{ijk}+N_2^{ijk}$ and using the symmetries 
\begin{equation}\label{N12-sym}
    N_1^{ijk}=N_1^{jik}\,,\qquad N_2^{ijk}=N_2^{kji}
\end{equation}
 we derive from (\ref{ABC-check-N}) 
\begin{equation}\label{ABC-check-N2}
    \check{B}_i{}^m=\epsilon_{ijk}N_1^{kmj}\,,\qquad 
    \check{C}_i{}^m=\epsilon_{ijk}N_2^{jkm}\,.
\end{equation}
of 16 linear equations in 16 independent variables. 
Moreover, our system is decomposed into a pair of 8  independent systems in 8 independent variables. 
A general form of a solution of the system $\check{B}_i{}^m=\epsilon_{ijk}N_1^{kmj}$ can be written as 
\begin{equation}
    N_1^{kmj}= x%\a 
    \check{B}_p{}^k\epsilon^{pmj}+y%\b 
    \check{B}_p{}^m\epsilon^{pkj}+z% \gamma
    \check{B}_p{}^j\epsilon^{pmk}\,
\end{equation}
with unknown numerical coefficients $x,y,z.$ %$\a,\b,\gamma$. 
Applying the symmetry of $N_1^{kmj}$ from  (\ref{N12-sym}), we obtain  $z=0\,,x=y $. 
%$\a=\b\,, \gamma=0$.
Thus the most general solution of the first equation in (\ref{ABC-check-N2}) reads
\begin{equation}
    N_1^{kmj}= %\a
    x\left(\check{B}_p{}^k\epsilon^{pmj}+  \check{B}_p{}^m\epsilon^{pkj}\right)\,.
\end{equation}
Contracting both sides of this equation with  $\epsilon_{ijk}$, we derive
\begin{eqnarray}
    \check{B}_i{}^m=x\left(\check{B}_p{}^k\epsilon_{ijk}\epsilon^{pmj}+  \check{B}_p{}^m\epsilon_{ijk}\epsilon^{pkj}\right)\,.
\end{eqnarray}
Due to the rules (\ref{epsilon1},\ref{epsilon2}) we have here 
\begin{eqnarray}
    \check{B}_i{}^m=x\left(\check{B}_p{}^k(-\d^p_i\d^m_k+\d^m_i\d^p_k)-\check{B}_p{}^m\d^p_i\right)=-2x \check{B}_i{}^m \,.
\end{eqnarray}
Consequently, $x=-(1/2)$ and 
\begin{equation}
    N_1^{kmj}= -\frac 12  \left(\check{B}_p{}^k\epsilon^{pmj}+ \check{B}_p{}^m\epsilon^{pkj}\right)\,.
\end{equation}
In the same fashion, we derive the solution of the second system in (\ref{ABC-check-N2})
\begin{equation}
    N_2^{kmj}= -\frac 12  \left(\check{C}_p{}^k\epsilon^{pmj}+  \check{C}_p{}^j\epsilon^{pmk}\right)\,.
\end{equation}
%@@@OLD??
%\begin{equation}
%    N_2^{kmj}= -\frac 12  %\left(\check{C}_p{}^m\epsilon^{pkj}-  %\check{C}_p{}^j\epsilon^{mkp}\right)\,.
%\end{equation}
Finally,
\begin{equation}\label{N-BC}
N^{kmj}= -\frac 12 \left((\check{B}_p{}^k\epsilon^{pmj}+  \check{B}_p{}^m\epsilon^{pkj})+
(\check{C}_p{}^k\epsilon^{pmj}+  \check{C}_p{}^j\epsilon^{pmk})
\right)\,.
\end{equation}
%%\begin{equation}\label{N-BC}
%    N^{kmj}= -\frac 12  %\left((\check{B}_p{}^k\epsilon^{pmj}-  %\check{B}_p{}^m\epsilon^{kpj})+
%    (\check{C}_p{}^m\epsilon^{pkj}-  %\check{C}_p{}^j\epsilon^{mkp})
%%\end{equation}
As a result, the $SL(3,\mathbb R)$-structure on the basis vector space $V$ provides an additional specification of the 3-rd order tensor subspaces. In particular:
\begin{itemize}
    \item The totally symmetric subspace $\bf S$  is not sensitive to the additional $SL(3,\mathbb R)$-structure.
    \item The totally skew-symmetric subspace $\bf A$ is characterized by a unique parameter $A$.
    It is a pseudo-scalar under general linear transformations and a proper scalar under transformations from the group  $SL(3,\mathbb R)$.
    \item The mixed-symmetry subspace $\bf N$ of dimension 16 is characterized by two independent traceless matrices of 8 components each. These matrices transform as pseudo-tensors under  general linear transformations and as proper tensors under transformations from the group  $SL(3,\mathbb R)$. 
    The choice of these matrices is not unique. This fact is in a correspondence to the non-uniqueness  of the decomposition $N^{ijk}=N_1^{ijk}+N_2^{ijk}$ that was derived at the level of $GL(3,\mathbb R)$-structure.
\end{itemize}

\subsection{$SO(3,\mathbb R)$-decomposition}
The $SO(3,\mathbb R)$-structure is defined by proper orthogonal transformations that belong to the intersection of the special linear group and the orthogonal group, i.e., 
\begin{equation}
    SO(3,\mathbb R)=SL(3,\mathbb R)\cap O(3,\mathbb R).
\end{equation}
In this case, the metric tensor $g_{ij}$ and the permutation tensor $\epsilon_{ijk}$ can be used simultaneously for  decomposition of a 3-rd order tensor $T^{ijk}$. Recall that the  $SL(3,\mathbb R)$ group itself does not define an additional decomposition of the tensor space. Consequently, the $SO(3,\mathbb R)$ structure can provide only an alternative description of the subspaces that already derived by the $O(3,\mathbb R)$-decomposition.

The tensor $\epsilon_{ijk}$ is not relevant to the fully  symmetric subspace  ${\bf S}$. Consequently, it's decomposition remains the same  as it is given in  Eq.(\ref{S-decomp0}). 
Briefly this representation can be written as a pair of a vector and a fully traceless tensor
\begin{equation}
    S^{ijk}=\{\a^i, R^{ijk}\},
\end{equation}
with the dimension reduction $10=3+7$. 

The skew-symmetric subspace ${\bf A}$ of dimension 1 is represented in $SO(3,\mathbb R)$-framework by a pseudo-scalar $A$, namely $A^{ijk}=A\epsilon^{ijk}$. 

The mixed-symmetric tensor $N^{ijk}$ is represented in the $SL(3,\mathbb R)$-framework by two mixed-type tensors, say $\check{B}_i{}^j$ and $\check{C}_i{}^j$. These tensors are traceless, $\check{B}_i{}^i=\check{C}_i{}^i=0$, thus their  additional decomposition is impossible. 

When the $O(3,\mathbb R)$ group acts in addition to $SL(3,\mathbb R)$, the pseudo-tensors $\check{B}_i{}^j$ and $\check{C}_i{}^j$ can be inverted into covariant (or contravariant) 2-nd order  pseudo-tensors. It means that one can define 
\begin{equation}
    \check{B}_{ij}:=g_{jm}\check{B}_i{}^m\quad   {\rm and}\quad   \check{C}_{ij}:=g_{jm}\check{C}_i{}^m\, 
\end{equation}
that are completely equivalent to $\check{B}_i{}^j$ and $\check{C}_i{}^j$.
Now, the pseudo-tensors $\check{B}_{ij}$ and $\check{C}_{ij}$ can be irreducible decomposed into the sum of their symmetric and skew-symmetric parts 
\begin{equation}
     \check{B}_{ij}=\check{B}_{(ij)}+\check{B}_{[ij]}\quad   {\rm and}\quad   \check{C}_{ij}=\check{C}_{(ij)}+\check{C}_{[ij]}\,.
\end{equation}
Since the tensors are traceless, the dimensions of the corresponded parts are $8=5+3$, respectively. 
We denote the symmetric parts as 
\begin{equation}
    E_{ij}=\check{B}_{(ij)}\quad   {\rm and}\quad   F_{ij}=\check{C}_{(ij)}\,.
\end{equation}
As about the skew-symmetric parts of pseudo-tensors, they can be   expressed   as  vectors. Indeed, we can define 
\begin{equation}
    \epsilon^{ijk}\check{B}_{[ij]}\quad   {\rm and}\quad  \epsilon^{ijk}\check{C}_{[ij]}\,.
\end{equation}
Notice, that being defined as contractions of two pseudo-tensors, these two vectors are proper. They can be expressed as linear combinations of the vectors $\b^k$ and $\gamma^k$ defined above. 
Consequently, the  mixed-type parts are decomposed irreducible as 
\begin{equation}
    N_1^{ijk}=\{E_{ij},\b^i\} \quad   {\rm and}\quad N_2^{ijk}=\{F_{ij},\gamma^i\}\,.
\end{equation}
We will provide   the explicit expressions for these relations in the sequel. 
\subsection{Results:}
\begin{figure}
\begin{center}
\includegraphics[width=0.95\linewidth]{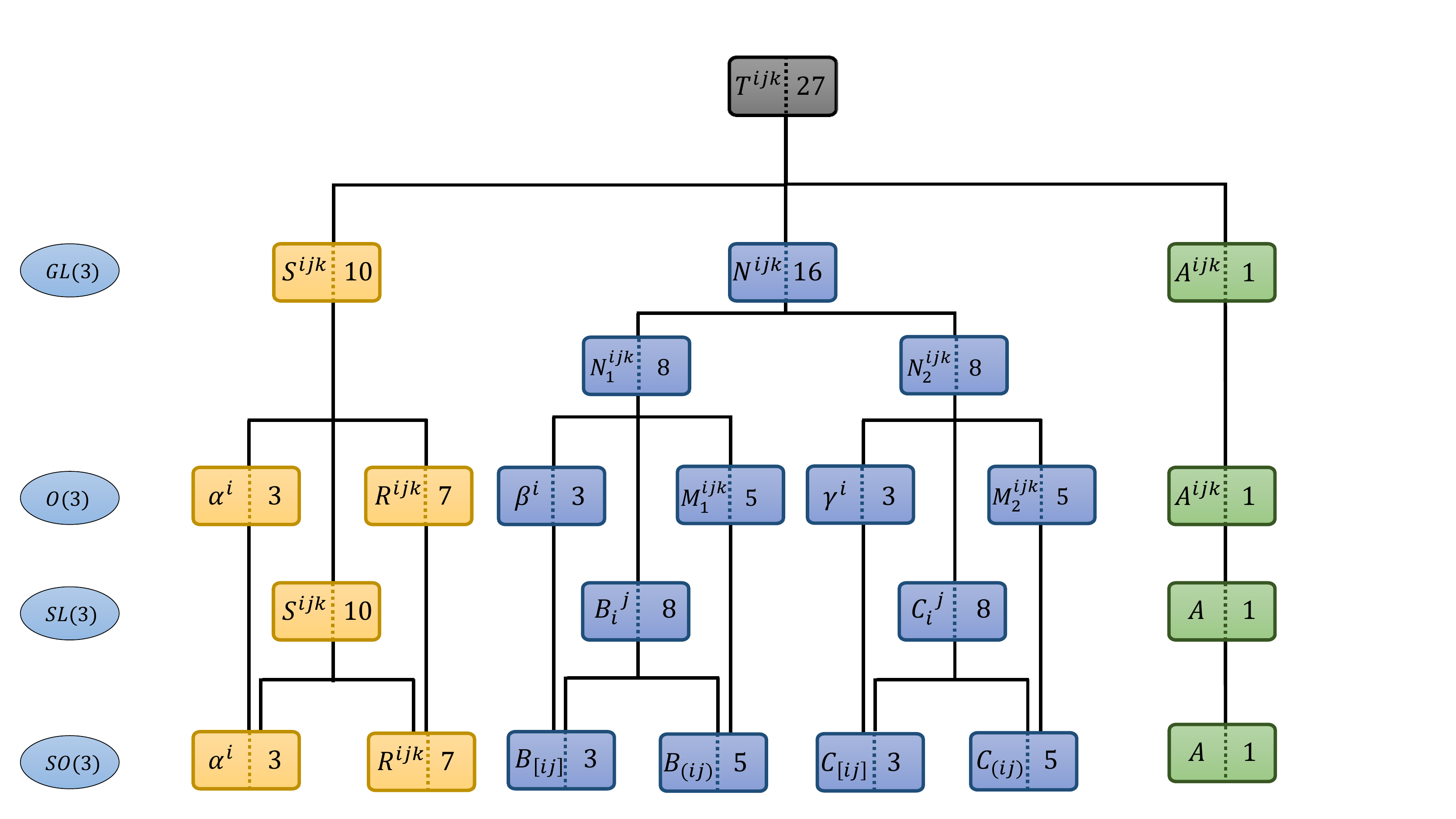}
\caption{Schematic representation of different types of decomposition of a generic 3-rd order  tensor. }
\label{figure:curve2}
\end{center}
\end{figure}
In this section we derived different types of decomposition of a generic 3-rd order tensor, see Fig.~1 for a  schematic representation. It is due to different structures on the vector space $V$ and correspondingly to different groups of transformations on $V$.

Our results are as follows: 
\begin{itemize}
    \item On the $GL(3,\mathbb R)$-level, the decomposition is presented uniquely  by three 3-rd order tensors. A successive decomposition into four 3-rd order tensors is irreducible but non-unique. Some explicit examples of equivalent irreducible decompositions are presented above.   
    \item  On the $O(3,\mathbb R)$-level with the metric tensor at hand, we can extract three vectors and 3 traceless 3-rd order tensors. These quantities provide the direct sum decomposition of the tensor space into 7 subspaces. 
    \item The $SL(3,\mathbb R)$-structure with a permutation tensor $\epsilon_{ijk}$ defined allows us to express the totally skew-symmetric tensor by a pseudo-scalar. Moreover, we can define mixed 2-nd order pseudo-tensors as an alternative representation of the  tensor $N^{ijk}$.
    \item The  $SO(3,\mathbb R)$-structure allows us to join the previous cases and yields the most finer decomposition. In particular,  it adds a representation of the tensor $N^{ijk}$ by three symmetric traceless 2-nd order tensors. Notably, the latter tensors have pseudo-tensorial nature.
\end{itemize}

\section{Tensors with partial symmetries and their decomposition}

In solid-state physics,  the 3-rd order tensors emerge as constitutive tensors. They establish linear phenomenological relations between the primary physical variables, \cite{Tinder}. Due to the fundamental symmetries of these variables,  constitutive tensors turn out to be partially symmetric. %For 3-rd tensors, there are two possible restrictions of  symmetry and skew-symmetry in a pair of indices. 
%Notice a general fact: For partially symmetric tensors, the tensor spaces cannot be reduced into smaller ones. Indeed,  some of the subspaces are completely violating, while others remain without any change. 

\subsection{Tensors with a symmetry  in a pair of indices: Piezoelectric tensor}
\subsubsection{Definition} 
Most phenomenological models used in solid-state physics are  dealing with the constitutive tensor that is symmetric in a pair of its indices. The well-known example is the piezoelectric tensor, see e.g., \cite{Haus}.   
This  tensor, $D^{ijk}$, is defined as a set of coefficients that relate the induced electric displacement vector $P^i$ to the second-order elasticity stress tensor $\sigma_{jk}$,
 \begin{equation}
     P^i=D^{ijk}\sigma_{jk}\,.
 \end{equation}
 Due to the symmetry of the stress tensor $\sigma_{ij}=\sigma_{ji}$, the piezoelectric tensor satisfies the symmetry relation
 \begin{equation}\label{p-sym}
     D^{ijk}=D^{ikj}\,.
 \end{equation}
In general, such a pair-symmetric tensor has 18 independent components. The invariant decomposition of the piezoelectric tensor can provide a useful information about the piezoelectric phenomena as well as about the proper classification of the piezoelectric crystals. Moreover, it can serve as a  useful theoretical tool for design novel  piezoelectric materials with specific  properties. Recently the algebraic properties of the pair-symmetric third-order tensor, especially in the case of the piezoelectric tensor, was studied intensively, see \cite{Qi}, \cite{} @@@@@@. Let us look at how the decomposition of the piezoelectric tensor $D^{ijk}$ can follow from the decomposition of the general tensor $T^{ijk}$ described above. 

\subsubsection{ $GL(3,R)$-decomposition} 
We start with the permutation decomposition of the piezoelectric tensor. Due to the symmetry (\ref{p-sym}), $D^{ijk}$ does not contain a totally skew-symmetric part $A^{ijk}$. Moreover, as the identity  $T^{(ijk)}=T^{(i(jk))}$ holds for an arbitrary tensor, the subspace  corresponding to the totally symmetric part of $D^{ijk}$ lies into  the tensor space defined by the tensor $D^{ijk}$. Consequently, the decomposition 
\begin{equation}\label{D-dec}
    D^{ijk}=S^{ijk}+N^{ijk}\,,\qquad {\rm where}\qquad S^{ijk}=D^{(ijk)}\,
\end{equation}
provides a unique decomposition of the space ${\bf D}$  into the direct sum of two subspaces with the indicated dimensions 
\begin{equation}\label{D-dec-1}
    {\mathbf {D}}={\bf S}\oplus{\bf N}\,,\qquad 18=10+8\,.
\end{equation}
Here the fully symmetric part  is reduced to 
\begin{equation}\label{S-piezo}
    S^{ijk}=D^{(ijk)}=
    \frac{1}{3}\left(D^{ijk}+D^{jki}+D^{kij}\right)\,,
\end{equation}
while the residue part reads
\begin{equation}\label{N-piezo}
    N^{ijk}=\frac{1}{3}\left(2D^{ijk}-D^{kij}-D^{jki}\right)\,.
\end{equation}
Notice that now $N^{ijk}=N^{ikj}$, i.e., this sub-tensor inhabits the symmetry of the piezoelectric tensor. In terms of tensor spaces, it means that space ${\bf N}$ corresponding to $N^{ijk}$ is a proper subspace of the tensor space ${\bf D}$, as it is indicated in Eq.(\ref{D-dec}). 

In contrast to the general case, the successive decomposition of the tensor  $N^{ijk}$ is impossible. Indeed, for a  general 3-rd order tensor $T^{ijk}$, the space  ${\bf N}$ is of dimension 16 and it is decomposed (non-uniquely) into the direct sum of two subspaces  ${\bf N_1, N_2}$ both of dimension 8. 
 However, in our current restricted case, the whole space ${\bf N}$  is only of dimension 8. So it cannot contain two disjoint proper subspaces of dimension 8. Let us check what is going with the subspaces of $N_{ijk}$ in this restricted case.  
First we observe that the Young tableau generated  subspaces $N_1^{ijk}$ and $N_2^{ijk}$ remain with the dimensions 8 also in our pair-symmetric case. But they do not have the fundamental symmetries (\ref{p-sym}), so their tensor spaces  are not subspaces of ${\bf N}$. In other words, even the equation $N^{ijk}=N_1^{ijk}+N_2^{ijk}$ holds,  ${\bf N}\ne{\bf N_1}\oplus{\bf N_2}$. The description is simplified  when we project the equation above onto the subsbace ${\bf D}$. It means that we  rearrange this equation into the form   $N^{i(jk)}=N_1^{i(jk)}+N_2^{i(jk)}$. For a pair-symmetric tensor,  the equality  $N_1^{i(jk)}=N_2^{i(jk)}$ holds, so we indeed do not have here any additional  decomposition. 

An alternative decomposition $N^{ijk}=\tilde{N}_1^{ijk}+\tilde{N}_2^{ijk}$ is even more suitable in our case. Indeed, now both sub-tensors $\tilde{N}_1^{ijk}$ and $\tilde{N}_2^{ijk}$ do have the fundamental symmetries  (\ref{p-sym}) and produce corresponded subspaces. However in the pair-symmetric case, $\tilde{N}_2^{ijk}=0$ and $N^{ijk}=\tilde{N}_1^{ijk}$. Once more we do not have any invariant decomposition of the tensor $N^{ijk}$. 

These facts demonstrate that   the decomposition (\ref{D-dec}) is the unique irreducible $GL(3,\mathbb R)$-invariant decomposition of the piezoelectric tensor.

\subsubsection{ $O(3,\mathbb R)$-decomposition}  Let us turn now to the decomposition of $D^{ijk}$ under the orthogonal group $O(3,\mathbb R)$. With  a metric tensor  $g_{ij}$ at hand we can construct from  a  third-order tensor $D^{ijk}=D^{i(jk)}$ only two trace-type vectors $v^k,w^k$ such that:
\begin{equation}\label{eq9}
v^k=g_{ij}D^{ijk}=g_{ij}D^{ikj}\qquad {\rm and}\qquad 
w^k=g_{ij}D^{kij}.
\end{equation}
How these vectors contribute to the successive decomposition of the piezoelectric tensor? 

%We start with the first part $S^{ijk}$. 
Due to the fully symmetry of  $S^{ijk}$, we have 
\begin{equation}
    g_{ij}S^{ijk}=g_{ik}S^{ijk}=g_{jk}S^{ijk}.
\end{equation}
 Thus  only one trace vector can be constructed from $S^{ijk}$. We denote it as 
\begin{equation}
   \alpha^k:= g_{ij}S^{ijk}=\frac{1}{3}g_{ij}\left(D^{ijk}+D^{jki}+D^{kij}\right)=\frac{1}{3}(2v^k+w^k)
\end{equation}
Consequently, the tensor $S^{ijk}$ is decomposed into the sum of two independent symmetric parts: 
\begin{equation}\label{S-decomp}
    S^{ijk}=K^{ijk}+R^{ijk}
\end{equation} 
  where the vector part is given in (\ref{S-decomp1}) by 
  \begin{equation}
      K^{ijk}=\frac{1}{5}\left(\alpha^i g^{jk} +\alpha^j g^{ik}+ \alpha^k g^{ij}\right),
  \end{equation}
 while  the residue  part  $R^{ijk}$ satisfies the traceless relations
\begin{equation}
  g_{ij}R^{ijk}=0,  \quad g_{ik}R^{ijk}=0,   \quad  g_{jk}R^{ijk}=0. 
\end{equation}
Hence, the tensor space ${\bf S}$ is decomposed into the direct sum of two subspeces with the  respective dimensions
 \begin{equation}
    {\bf S}={\bf K}\oplus {\bf R}\,, \qquad 10=3\oplus 7.
\end{equation}

Now we consider the trace decomposition  of the tensor $N^{ijk}$. Two possible vectors 
\begin{equation}
    g_{ij}N^{ijk}=\frac{1}{3}(v^k-w^k),\qquad g_{jk}N^{ijk}=\frac{2}{3}(w^i-v^i)\,
\end{equation}
turn out to be linearly dependent.
Consequently we have   here only one independent vector
\begin{equation}
    \beta^j=\frac{2}{3}(v^j-w^j)
\end{equation}
Thus the tensor $N^{ijk}$ can be decomposed into the sum:
\begin{equation}
    N^{ijk}=M^{ijk}+P^{ijk}\,,
\end{equation}
 where the vector part
% \begin{eqnarray}
 %    M^{ijk}&=&\frac{1}{3}g^{jk}(w^i-v^i)+\frac{1}{6}g^{i%k}(v^j-w^j)+\frac{1}{6}g^{ij}(v^k-w^k)\nonumber\\
 %    &=&\frac{1}{3}g^{jk}(w^i-v^i)+\frac{1}{3}g^{ik}(v^j-%w^j)    \nonumber\\ %&=&\frac{1}{2}(g^{ik}\beta^j-g^{jk}\beta^i)\,.
 %\end{eqnarray}
 \begin{equation}\label{symM}
   M^{ijk}=  \frac{1}{4}\left(2g^{jk}\beta^i-g^{ik}\beta^j-g^{ij}\beta^k\right)\,,
 \end{equation}
while the residue part $P^{ijk}=N^{ijk}-M^{ijk}$ satisfies the traceless relation 
\begin{equation}
     g_{ij}P^{ijk}=0,  \quad g_{ik}P^{ijk}=0,   \quad  g_{jk}P^{ijk}=0\,.
\end{equation}
Accordingly, we have a  decomposition of the  subspace ${\bf N}$ into the direct sum of two invariant subspaces with the corresponding  dimensions 
\begin{equation}
    {\bf N}={\bf M}\oplus {\bf P}\,, \qquad 8=3\oplus 5\,.
\end{equation}

 \subsubsection{ $SL(3,\mathbb R)$-decomposition} 
 In the $SL(3,\mathbb R)$ background, the permutation pseudotensor $\epsilon_{ijk}$ is available. Since the fully antisymmetric part of the piezoelectric tensor vanishes, the pseudo-scalar $A=0$. 
 Moreover, the contraction in two indices of the symmetric part $S^{ijk}=D^{(ijk)}$ with $\epsilon_{ijk}$ is trivial as in the general case. Consequently,  $\epsilon_{ijk}$ acts only on the mixed-type tensor $N^{ijk}$. 
  For the 2-nd order tensors defined in Eq.(\ref{ABC}), we have 
 \begin{equation}\label{ABC1}
    A_i{}^m=\epsilon_{ijk}D^{mjk}=0\,,
    \qquad 
    B_i{}^m=\epsilon_{ijk}D^{kmj}=\epsilon_{ijk}D^{jkm}=\ 
    -C_i{}^m\,.
\end{equation}
Since these tensors are traceless and satisfy the relations  $A_i{}^m+B_i{}^m+C_i{}^m=0$, there is only one nontrivial contraction, say $B_i{}^m$. 
In order to derive the contribution of this tensor to   $N^{ijk}$, we start with the  linear relations between these two tensors of 8 independent components
\begin{equation}
   B_i{}^m=\epsilon_{ijk}D^{kmj}=\epsilon_{ijk}{N}^{kmj}\,.
\end{equation}
To inverse this relation, we substitute    $C_i{}^m=-B_i{}^m$ into Eq.(\ref{N-BC}). Consequently,  we have
\begin{equation}\label{N-B}
    N^{kmj}=\frac 12  \left( {B}_p{}^m\epsilon^{kpj}+
 {B}_p{}^j\epsilon^{kpm}\right)\,.
\end{equation}
Observe that this expression has a desired  symmetry $N^{kmj}=N^{kjm}$.
As a result, the representations of the mixed-symmetry part of the piezoelectric tensor by the 3-rd order tensor $N^{ijk}$ and by the 2-nd order pseudo-tensor ${B}_i{}^j$ are completely equivalent. We are left with the decomposition of the piezoelectric tensor space into two subspaces with the dimensions $18=10+8$. %We write briefly %\begin{equation}
 %   D^{ijk}\to \{S^{ijk}, B_i{}^j\}
 %   \end{equation}
 
 \subsubsection{ $SO(3,\mathbb R)$-decomposition} 
 On the $SO(3,\mathbb R)$-level, the mixed pseudo-tensor $B_i{}^j$ can be reverted into the contravariant tensor $B_{ij}=g_{jm}B_i{}^m$. This new tensor can be decomposed into the skew-symmetric and symmetric parts with the dimensional reduction $8=3+5$. Let us show that this is exactly the same decomposition that we already have on the $O(3,\mathbb R)$-level. In particular,  we can show that the skew-symmetric part $B_{[ij]}\sim \b^i$, while the symmetric part   $B_{(ij)}\sim P^{ijk}$.

 Substituting into Eq.(\ref{N-B}) the expression ${B}_p{}^m=g^{mr}B_{pr}$ we obtain \begin{equation}\label{N-B2}
    N^{kmj}=\frac 12  \left( g^{mr}\epsilon^{kpj}+g^{jr}
 \epsilon^{kpm}\right)B_{pr}=\frac 12  \left( g^{mr}\epsilon^{kpj}+g^{jr}
 \epsilon^{kpm}\right)\left({B}_{[pr]}+{B}_{(pr)}\right)\,.
\end{equation}
In order to show that the first term gives the tensor $M^{ijk}$, we chose the representation, 
\begin{equation}
    {B}_{[pr]}=\frac 12\epsilon_{prs}\b^s\,.
\end{equation}
Then the first term in (\ref{N-B2}) is expressed as 
\begin{eqnarray}
    &&\frac 12  \left( g^{mr}\epsilon^{kpj}+g^{jr}
 \epsilon^{kpm}\right){B}_{[pr]}=\frac 14 \left( g^{mr}\epsilon^{kpj}+g^{jr}
 \epsilon^{kpm}\right)\epsilon_{prs}\b^s=\nonumber\\
 && -\frac 14\left(g^{mr}(\d^k_r\d^j_s-\d^k_s\d^j_r)+g^{jr}(\d^k_r\d^m_s-\d^k_s\d^m_r)\right)\b^s=\nonumber\\
 &&\frac 14\left(2g^{jm}\b^k-g^{mk}\b^j-
 g^{jk}\b^m\right)=M^{ijk}\,.
\end{eqnarray}
Hence we identified the vector part of 3 independent components (\ref{symM}) with the first term of  (\ref{N-B2}). The residue part is traceless and expressed by 5  independent components ${B}_{(ij)}$. So it is equal to $P^{ijk}=N^{ijk}-M^{ijk}$. 

In the decomposition $N^{ijk}=M^{ijk}+P^{ijk}$, we derived the representations of the irreducible parts $M^{ijk}$ and $P^{ijk}$ by two independent matrices: The  antisymmetric matrix $B_{[ij]}$ and the symmetric traceless matrix $B_{(ij)}$. Explicitly, 
\begin{eqnarray}
    M^{ijk}&=&\frac 12  \left( g^{mr}\epsilon^{kpj}+g^{jr}
 \epsilon^{kpm}\right){B}_{[pr]}\,,\\
 P^{ijk}&=&\frac 12  \left( g^{mr}\epsilon^{kpj}+g^{jr}
 \epsilon^{kpm}\right){B}_{(pr)}\,.
\end{eqnarray}
%{\color {red}Compare \cite{Jerphagnon}.}

 \subsubsection{Results}
 Let us summarize our results on the irreducible decomposition of the pieszoelectric tensor (and tensors with similar pair symmetry). For a schematic   representation, see Fig.2. Notice that our approach allows to distinguish the vectors $\a^i$ and $\b^i$. These two vectors are merely coming from two different irreducible parts ${\bf S}$ and ${\bf N}$.
  \begin{figure}
\begin{center}
\includegraphics[width=0.8\linewidth]{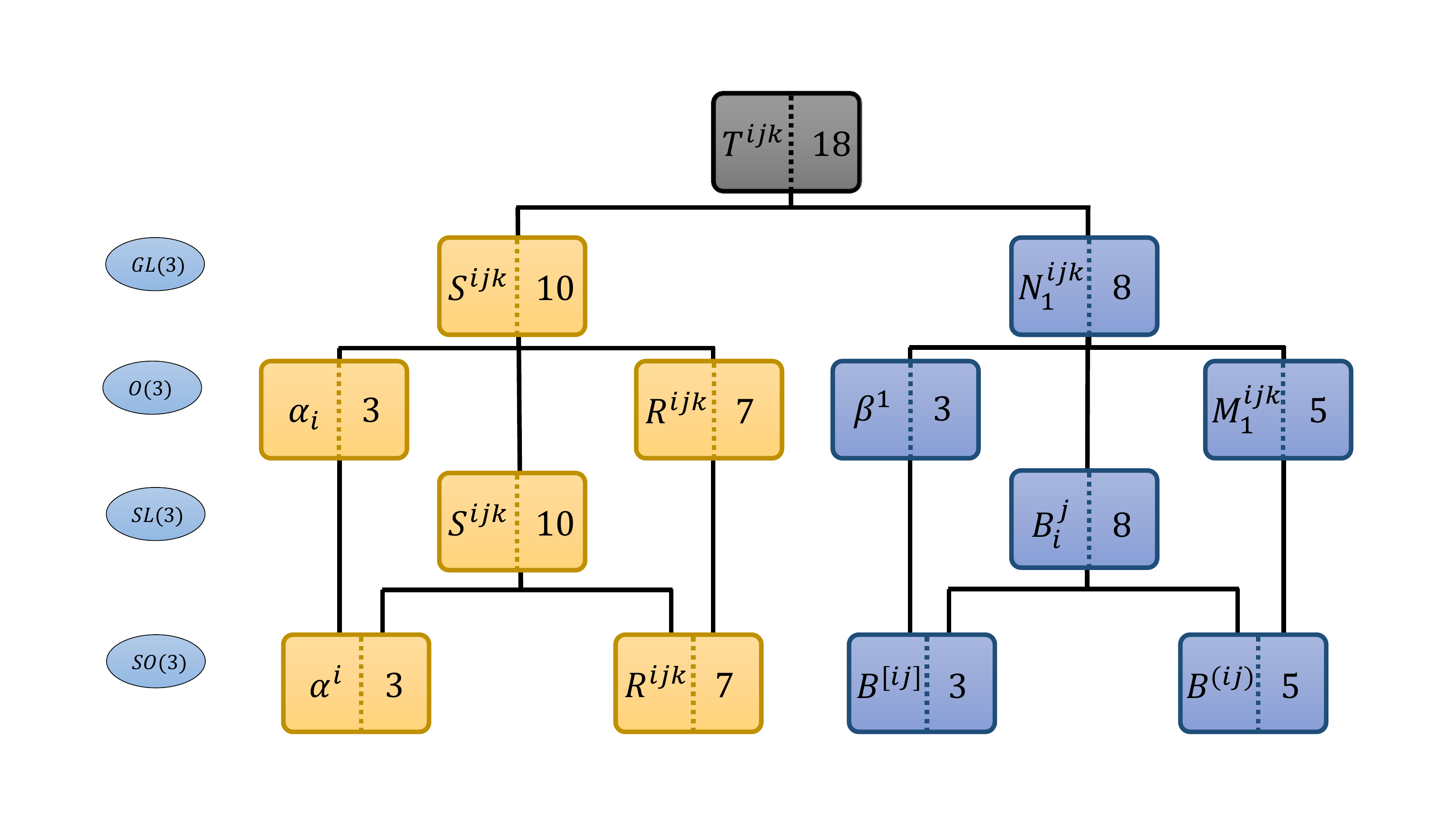}
\caption{Schematic representation of different levels of decomposition of a  3-rd order  piezoelectric tensor. }
\label{figure:curve2a}
\end{center}
\end{figure}

 \begin{prop}
 Let a piezoelectric-type tensor $D^{ijk}$ with the symmetry $D^{ijk}=D^{ikj}$ be given.
 \begin{itemize}
     \item On the $GL(3,\mathbb R)$-level, the tensor is decomposed uniquely and irreducibly into the sum of two subtensors (\ref{S-piezo}) and (\ref{N-piezo}). The tensor space is decomposed into the direct sum of two subspaces: 
     \begin{equation}
         {\bf D}={\bf S}\oplus{\bf N}\,,\qquad 18=10+8\,.
     \end{equation}
     \item On the $O(3,\mathbb R)$-level, the fully symmetric part $S^{ijk}$ and the mixed part $N^{ijk}$ are  decomposed both into the vector (trace)  and the traceless parts. In term of subspaces, 
     \begin{equation}
          {\bf S}={\bf K}\oplus{\bf R}\,,\qquad 10=3+7\,.
     \end{equation}
     \begin{equation}
          {\bf N}={\bf M}\oplus{\bf P}\,,\qquad 8=3+5\,.
     \end{equation}
     Consequently,
     \begin{equation}
         {\bf D}=({\bf K}\oplus{\bf R})\oplus({\bf M}\oplus{\bf P})\,,\qquad 18=(3+7)+(3+5)\,.
     \end{equation}
     \item The subspaces ${\bf K},{\bf R},{\bf M},$ and ${\bf P}$ are mutually orthogonal one to another. 
     \item On the $SL(3,\mathbb R)$-level, the tensor $N^{ijk}$ is expressed as a pseudo-tensor $B^i{}_j$.
      \item On the $SO(3,\mathbb R)$-level, the fully symmetric tensor $S^{ijk}$ returns to the  $O(3,\mathbb R)$-decomposition that is given by the vector and traceless parts. 
      The pseudo-tensor $B^i{}_j$ is decomposed irreducibly and uniquely into the sum of two pseudo-tensors $B_{[ij]}$ and $B_{(ij)}$ that span the spaces ${\bf M}$ and ${\bf P}$, respectively. 
 \end{itemize}
 \end{prop}

 \subsection{Tensors with skew-symmetry in a pair of indices: Hall tensor}
 \subsubsection{Definition}
 An alternative  constraint of the general third-order constitutive tensor emerges in the description of the Hall effect and  the Faraday effect. In this case, the constitutive tensor is skew-symmetric in a pair of its indices. In particular, we consider here the Hall effect that is widely  observed in conductors and semi-conductors. In  anisotropic media, it is described  as follows,see, e.g.,\cite{Haus} ,\cite{Liu}
  \begin{equation}
     E_i=\kappa_{ijk}I^jH^k\,,
 \end{equation}
 where $I^j$ is an electric current density, $H^k$ is a magnetic field, while $E_i$ is an induced electric field.  For physical details, see  e.g. \cite{Haus}. 
 The set of the coefficients  $\kappa_{ijk}$ describes physical properties of the conductor. These coefficients  form a contravariant third-order tensor  referred to as the {\it Hall tenor.}
 Due to statistical arguments, the  Onsager relations   yield the basic symmetry of the Hall tensor \cite{LL8}  
 \begin{equation}
  \label{a-sym}   \kappa_{ijk}=-\kappa_{jik}\,.
 \end{equation}
In general, the skew-symmetric tensor  $\kappa_{ijk}$ has 9 independent components. The decomposition of this tensor was studied recently in various publications, see, e.g., \cite{Liu} and the references given therein.  Let us look how this decomposition follows from  the decomposition of the general third-order tensor (in the contravariant version).

 \subsubsection{ $GL(3,R)$-decomposition}
 
We start with the permutation decomposition of the Hall  tensor. Due to the skew-symmetry (\ref{a-sym}), this tensor does not contain the fully symmetric part $S_{ijk}$. Thus it is decomposed uniquely into a  sum of two independent parts: the fully skew-symmetric part $A_{ijk}$, and the mixed-symmetry part $N_{ijk}$.  This decomposition with the indicated dimensions reads
\begin{equation}\label{P-dec}
    \kappa_{ijk}=A_{ijk}+N_{ijk}\,,\qquad 9=1+8\,.
\end{equation}
Due to (\ref{a-sym}), the skew-symmetric part is reduced to 
\begin{equation}
    A_{ijk}=\kappa_{[ijk]}=
    \frac{1}{3}\left(\kappa_{ijk}+\kappa_{jki}+\kappa_{kij}\right)\,,
\end{equation}
while the residue mixed-symmetry part reads
\begin{equation}
    N_{ijk}=\kappa_{ijk}-\kappa_{[ijk]}=\frac{1}{3}\left(2\kappa_{ijk}-\kappa_{kij}-\kappa_{jki}\right)\,.
\end{equation}
Notice that $A_{ijk}=-A_{jik}$ and $N_{ijk}=-N_{jik}$, i.e., these partial tensors inhabit the skew-symmetry of the Hall  tensor.

 Observe that, similarly to the piezoelectric tensor, and in  a contrast to the general case, the successive decomposition of the tensor  $N^{ijk}$ is impossible. Indeed on the $GL(3,\mathbb R)$-level, a subspace of the dimension 8 is the smallest one. Then it cannot contain $GL(3,\mathbb R)$-invariant proper subspaces.

 The Young tableau generated  subspaces $N_{1ijk}$ and $N_{2ijk}$ do not vanish. However these tensors do not satisfy the skew-symmetry relation (\ref{a-sym}). Consequently,  they do not form  subspaces of the $\kappa$-space. 
 They projection however are properly embedded into this space. This projection can be provided by    antisymmetrization of the tensors in the first pair of their indices. In fact, we have $N_{1[ij]k}=0$ and $N_{2[ij]k}=N_{ijk}$. 
 
 The alternative decomposition $N_{ijk}=\hat{N}_{1ijk}+\hat{N}_{2ijk}$ is even more suitable in our case. Indeed, these sub-tensors  do have the skew-symmetries  (\ref{a-sym}). Thus they  produce   subspaces of the $\kappa$-space. However, in the case of the Hall tensor, $\hat{N}_{2ijk}=0$ and $\hat{N}_{1ijk}=N_{ijk}$.
 
 These facts demonstrate that   the decomposition (\ref{P-dec}) is a unique irreducible $GL(3,\mathbb R)$-invariant decomposition of the 
Hall tensor. 
 
 \subsubsection{ $O(3,R)$-decomposition} 
Let us consider now the decomposition of $\kappa_{ijk}$ under the orthogonal group $O(3,\mathbb R)$. With  a covariant version of the  metric tensor  $g^{ij}$ at hand we can construct from  $\kappa_{ijk}$  only one   trace-type covector. Indeed for three vectors  $u_k,v_k,w_k$ defined in Eq.(\ref{3vectors}), we have 
 \begin{equation}\label{Hvectors}
 u_k=g^{ij}\kappa_{ijk}=0\,,\qquad   v_k=g^{ij}\kappa_{ikj}=-g^{ij}\kappa_{kij}=-w_k\,.
\end{equation}
Consequently we have only one independent covector $v_k$ that contributes to the mixed-symmetry part $N_{ijk}$. Due to Eq.(\ref{M-tensor}), we use in the decomposition 
\begin{equation}\label{M+P-Htensor}
    N_{ijk}=M_{ijk}+P_{ijk}\,
\end{equation}
 the vector part 
\begin{equation}\label{M-tensorH}
    M_{ijk}=%\frac{1}{6}\left(  %%(2u^k-v^k-w^k) g^{ij}+
   % -3v_ig^{jk}   %%(2w^i-u^i-v^i) g^{jk}
   % + 3v_jg_{ik}%%(2v^j-u^j-w^j) g^{ik}
  %  \right)
  \frac{1}{2}(v_jg_{ik}-v_ig_{jk})\,.
\end{equation}
We can check straightforwardly that the residue part $P_{ijk}=N_{ijk}-M_{ijk}$  satisfies the traceless relation 
\begin{equation}
    g^{ij}P_{ijk}=g^{ij}P_{ikj}=g^{ij}P_{kij}=0\,.
\end{equation}
Accordingly, we have a direct sum decomposition of the  subspace ${\bf N}$ into the direct sum of two invariant subspaces with the corresponding  dimensions 
\begin{equation}
    {\bf N}={\bf M}\oplus {\bf P}\,, \qquad 8=3\oplus 5\,.
\end{equation}

 \subsubsection{ $SL(3,R)$-decomposition}
 With the volume element structure defined on the basic space $V$, we are able to define the pseudo-scalar 
 \begin{equation}
     A=\frac 16 \epsilon^{ijk}\kappa_{ijk}=\frac 16 \epsilon^{ijk}A_{ijk}\,.
 \end{equation}
 The inverse relation reads
 \begin{equation}
     A_{ijk}=A\epsilon_{ijk}
 \end{equation}
 Similarly to Eq.(\ref{ABC}) we define three pseudo-tensors 
 \begin{equation}\label{ABC-P}
    A^i{}_m=\epsilon^{ijk}\kappa_{mjk}\,,\qquad 
    B^i{}_m=\epsilon^{ijk}\kappa_{kmj}\,,\qquad 
    C^i{}_m=\epsilon^{ijk}\kappa_{jkm}\,.
\end{equation}
 Observe that
 \begin{eqnarray}
     B^i{}_m=\epsilon^{ijk}\kappa_{kmj}=-\epsilon^{ijk}\kappa_{mkj}=
     -\epsilon^{ikj}\kappa_{mjk}=\epsilon^{ijk}\kappa_{mjk}=A^i{}_m\,.
 \end{eqnarray}
 Then also for the modified pseudo-tensors we have  $\check{B}^i{}_m=\check{A}^i{}_m$. 
 We use the relation (\ref{ABC-check-sum}) to derive the second relation between the pseudotensors,  $\check{C}^i{}_m=-2\check{A}^i{}_m$. Now we are able to express the mixed-symmetry part of the Hall tensor in term of the pseudo-tensor $\check{A}^i{}_m$. Changing the position of the indices in Eq.(\ref{N-BC}) and substituting the relations above, we derive
 \begin{equation}\label{N-BC-H}
N_{kmj}= \frac 12 \left(\check{A}^p{}_k\epsilon_{pmj}-\check{A}^p{}_m\epsilon_{pkj}-2\check{A}^p{}_j\epsilon_{kmp}\right)\,.
%(\check{B}_p{}^k\epsilon^{pmj}-  \check{B}_p{}^m\epsilon^{kpj})+
%(\check{C}_p{}^k\epsilon^{pmj}-  \check{C}_p{}^j\epsilon^{kmp})
%\right)\,.
\end{equation}
 Notice that the lhs of this equation is skew-symmetric in the indices $k$ and $m.$ Eq.(\ref{N-BC-H}) proves that the representations of the mixed-symmetry part by the third-order tensor $N^{kmj}$ and the pseudo-tensor $\check{A}_p{}^k$ are completely identical. 
 
 \subsubsection{ $SO(3,R)$-decomposition}
 On a space endowed with the metric and the volume element structures, the mixed-type pseudo-tensor $\check{A}_i{}^j$ can be reverted into the covariant tensor. 
 Let us define
 \begin{equation}
     \check{A}^{ij}:=g^{mj}\check{A}^i{}_m\,.
 \end{equation}
 This traceless tensor of 8 independent components can be irredusibly decomposed into the sum of its symmetric and skew-symmetric parts of $8=3+5$ independent 
 components, respectively 
 \begin{equation}\label{N-BC-H1-0}
     \check{A}^{ij}=\check{A}^{[ij]}+\check{A}^{(ij)}\,.
 \end{equation}
 This decomposition  has to be in  a correspondence with the $O(3,\mathbb R)$-decomposition (\ref{M+P-Htensor}). In particular, we can guess the relations of the form $\check{A}^{[ij]}\sim v_i$ and $\check{A}^{(ij)}\sim P_{ijk}$. 
 Let us derive these expressions explicitly.
 Substituting into Eq. (\ref{N-BC-H}) the expression $\check{A}^p{}_m=\check{A}^{pr}g_{mr}$ we obtain 
 \begin{equation}\label{N-BC-H1}
N_{kmj}= \frac 12 \check{A}^{pr}\left(g_{kr}\epsilon_{pmj}-
g_{mr}\epsilon_{pkj}-2g_{jr}\epsilon_{kmp}\right)\,.
\end{equation}
Substituting here (\ref{N-BC-H1-0}) we have
\begin{equation}\label{N-BC-H1x}
N_{kmj}=\frac 12 \left(\check{A}^{[pr]}+\check{A}^{(pr)}\right)\left(g_{kr}\epsilon_{pmj}-
g_{mr}\epsilon_{pkj}-2g_{jr}\epsilon_{kmp}\right)\,.
\end{equation}
 We assume the skew-symmetric part of $\check{A}^{pr}$ to be parametrized as 
 \begin{equation}
     \check{A}^{[pr]}=-\frac 13 \epsilon^{prs}v_s\,.
 \end{equation}
 This expression is contributed to the right-hand-side of Eq.(\ref{N-BC-H1x}) as 
  \begin{equation}\label{N-BC-H2a}
-\frac 16  \epsilon^{prs}\left(g_{kr}\epsilon_{pmj}-
g_{mr}\epsilon_{pkj}-2g_{jr}\epsilon_{pkm}\right)v_s\,.
\end{equation}
Using (\ref{epsilon1}) we calculate this expression to  obtain that it is equal to the vector part of $N_{kmj}$
   \begin{equation}\label{N-BC-H2}
\frac 12 (g_{jk}v_m-g_{mj}v_k)=M_{kmj}\,.
\end{equation}
Thus we have proved that the skew-symmetric part of the pseudo-tensor, $\check{A}^{[pr]}$, indeed equivalent to the vector $v_k$. Since $\check{A}^{pr}$ contributes linearly to the right-hand-side of Eq.(\ref{N-BC-H1x}), the symmetric part of the  pseudo-tensor, $\check{A}^{(pr)}$ is equivalent to the traceless part $P_{ijk}$. 
 In particular, we have the identities 
  \begin{equation}\label{N-BC-H3}
M_{kmj}= \frac 12 \check{A}^{[pr]}\left(g_{kr}\epsilon_{pmj}-
g_{mr}\epsilon_{pkj}-2g_{jr}\epsilon_{kmp}\right)\,,
\end{equation}
 and 
   \begin{equation}\label{N-BC-H4}
P_{kmj}= \frac 12 \check{A}^{(pr)}\left(g_{kr}\epsilon_{pmj}-
g_{mr}\epsilon_{pkj}-2g_{jr}\epsilon_{kmp}\right)\,.
\end{equation}
 
 \subsubsection{Results}
 In this section, we derived the invariant decomposition of the Hall tensor for different categories  of the basic vector space. 
 We derived explicitly the  relation between the tensor and pseudo-tensor representations of the mixed-symmetry parts. The results are presented in Fig.3.  
  \begin{figure}[!ht]
\begin{center}
\includegraphics[width=0.8\linewidth]{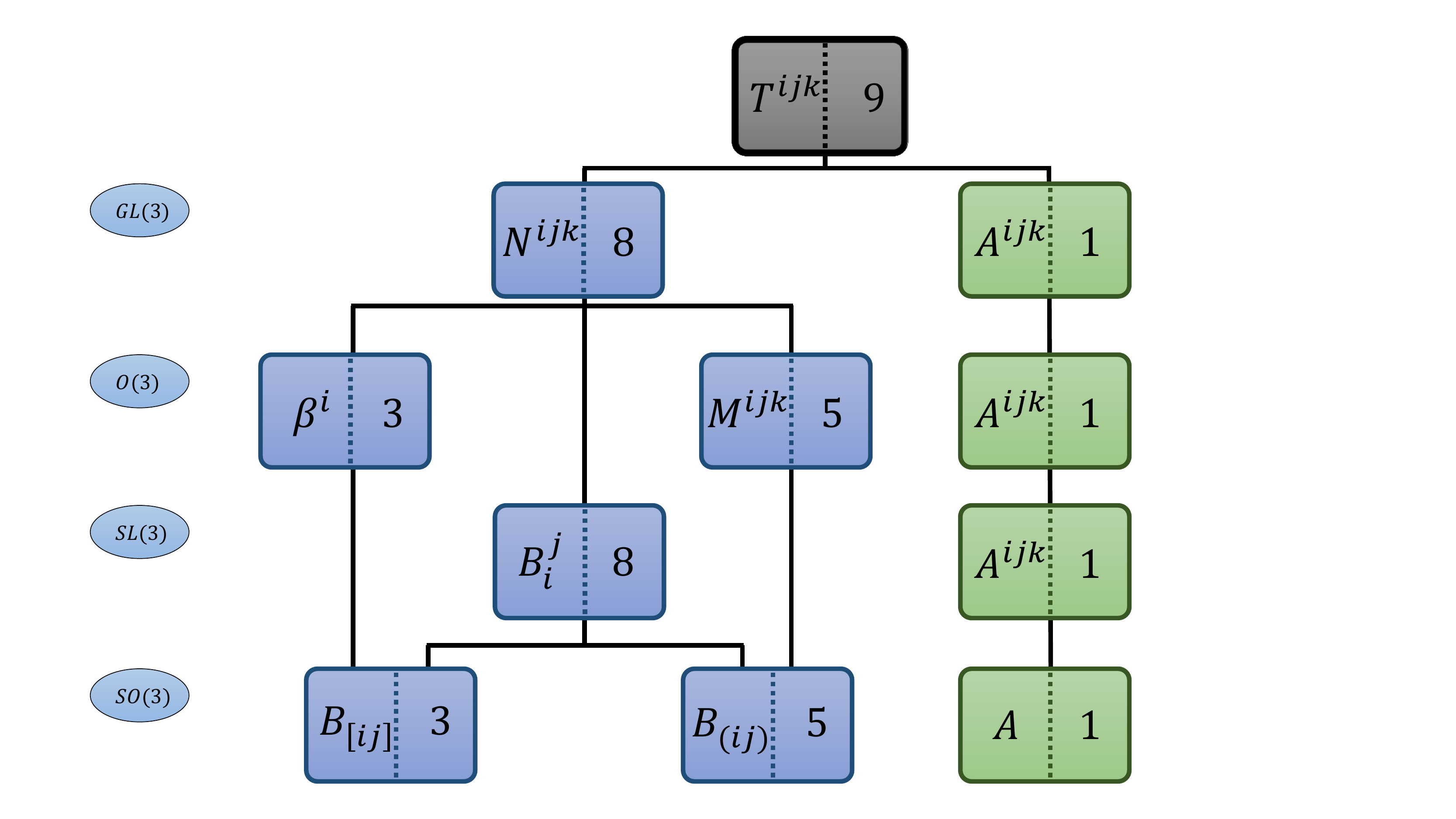}
\caption{Schematic representation of different types of decomposition of the  3-rd order Hall tensor. }
\label{figure:curve2b}
\end{center}
\end{figure}

 \begin{prop}
 Let a tensor $\kappa_{ijk}$ with the symmetry $\kappa_{ijk}=-\kappa_{jik}$ be given. 
 \begin{itemize}
     \item On the $GL(3,\mathbb R)$-level, the tensor is decomposed uniquely and irreducibly into the sum of two subtensors. The tensor space is decomposed into the direct sum of two subspaces: 
     \begin{equation}
         {\bf \kappa}={\bf N}\oplus{\bf A}\,,\qquad 18=10+8\,.
     \end{equation}
     \item On the $O(3,\mathbb R)$-level, the mixed- symmetry part $S^{ijk}$ is  decomposed  into the vector (trace)  and the traceless parts. In term of subspaces, 
     \begin{equation}
          {\bf S}={\bf M}\oplus{\bf P}\,,\qquad 8=3+5\,.
     \end{equation}
     Consequently,
     \begin{equation}
         {\bf \kappa}=({\bf M}\oplus{\bf P})\oplus{\bf A}\,,\qquad 9=(3+5)+1\,.
     \end{equation}
     \item The subspaces ${\bf M},{\bf P},$ and ${\bf A}$ are mutually orthogonal one to another. 
     \item On the $SL(3,\mathbb R)$-level, the tensor $N^{ijk}$ is expressed as a pseudo-tensor $B^i{}_j$.
      \item On the $SO(3,\mathbb R)$-level, the pseudo-tensor $B^i{}_j$ is decomposed irreducibly and uniquely into the sum of two pseudo-tensors $B_{[ij]}$ and $B_{(ij)}$. These tensors span the spaces ${\bf M}$ and ${\bf P}$, respectively. 
 \end{itemize}
 \end{prop}

\section{Conclusion}
In this paper, we study an invariant decomposition of the 3-rd order tensor into smaller sub-tensors. Even  this relatively simple case, demonstrates the principle problems of non-uniqueness of irreducible decomposition. We constructed explicitly different types of decomposition based on various geometrical structures defined on a basic vector space. 

In the cases of a physical interest, the 3-rd order tensor emerges as a constitutive tensor with additional symmetries coming from the symmetries of the basic physical variables. We considered the pair symmetric tensor of a piezoelectric type and a skew-symmetric tensor of a Hall type. We show that for such tensors the irreducible decomposition is unique without ambiguity.  

The problem of irreducible decomposition is not of the mathematical interest only. In fact, different irreducible invariant parts of the constitutive tensor have to represent different physical features of  media. The examples of such type of presentation is known in gravity, solid state electromagnetism, and elasticity theory. 

The irreducible decomposition presented here can be useful for investigation of the natural materials (as piezoelectric crystals) as well as for design the artificial nano-materials.

\section*{Data Availability Statement}
Data sharing is not applicable to this article as no new data were created or analyzed in this
study.

\bibliographystyle{plain}
%\bibliography{references}

\end{document}